\documentclass[11pt]{scrartcl}

\usepackage[utf8]{inputenc}
\usepackage[T1]{fontenc}

\usepackage[bottom=1in,top=0.85in,left=0.85in,right=0.85in]{geometry}\usepackage{times}
\RedeclareSectionCommand[indent=0pt]{subparagraph}

\usepackage{amsmath,amsthm,amssymb}
\usepackage{mathtools}
\usepackage[hyphens]{url}
\usepackage{hyperref}  
\usepackage{xspace}
\usepackage{thm-restate}

\usepackage{tabularx}
\newcolumntype{L}{>{\raggedright\arraybackslash}X}
\usepackage[ruled]{algorithm2e}
\usepackage{enumitem}

\usepackage{booktabs}
\newtheorem{theorem}{Theorem}
\newtheorem{definition}[theorem]{Definition}
\newtheorem{proposition}[theorem]{Proposition}
\newtheorem{lemma}[theorem]{Lemma}

\newtheorem{observation}[theorem]{Observation}
\newtheorem{corollary}[theorem]{Corollary}

\usepackage{tikz}
\usetikzlibrary{calc,shapes}

\tikzset{
  point/.style={draw, circle, fill=black, inner sep=0mm, minimum size=6pt},
  rp/.style={point, red!70!yellow},
  bp/.style={point, blue!50!black},
  cluster/.style={draw, gray!40!white, fill=gray!10!white, very thick},
  fairlet/.style={draw, gray!80!white, fill=gray!20!white, rounded corners, very thick}
}
\DeclareRobustCommand{\RedNode}{%
\begin{tikzpicture}
\draw (0,0) node[rp] {};
\end{tikzpicture}%
}
\DeclareRobustCommand{\BlueNode}{%
\begin{tikzpicture}
\draw (0,0) node[bp] {};
\end{tikzpicture}%
}

\newcommand{\srsum}[1]{\smashoperator[r]{\sum_{#1}}}

\newcommand{\slrsum}[1]{\smashoperator[lr]{\sum_{#1}}}

\DeclareMathOperator{\poly}{poly}

\newcommand{\set}[1]{\{#1\}}
\newcommand{\kk}{\texorpdfstring{$k$}{k}}
\newcommand{\OPT}{\text{OPT}}
\newcommand{\OPTfl}{\OPT_{\text{fl}}}
\newcommand{\OPTcb}{\OPT_{\text{cb}}}
\newcommand{\N}{\mathbb{N}}
\newcommand{\R}{\mathbb{R}}

\newcommand{\setc}[2]{\left\{#1\ |\ #2\right\}}

\newcommand{\ie}{i.e.\xspace}

\DeclareMathOperator{\balance}{balance}

\DeclareMathOperator{\dist}{d}

\DeclareMathOperator{\cost}{cost}
\DeclareMathOperator{\mass}{mass}

\DeclarePairedDelimiter{\norm}{\lVert}{\rVert}

\newcommand{\Bigset}[2]{\Bigl\{#1\mathrel{}\Bigm\lvert\mathrel{} #2\Bigr\}}

\usepackage{thm-restate}

\title{On the cost of essentially fair clusterings}
\author{}
\author{Ioana O. Bercea\thanks{School of Electrical Engineering, Tel Aviv University, Israel} \and Martin Gro\ss{}\thanks{School of Business and Economics, RWTH Aachen, Germany} \and Samir Khuller\thanks{Department of Computer Science, University of Maryland, College Park, USA} \and Aounon Kumar\footnotemark[3] \and Clemens R\"osner\thanks{Institute of Computer Science, University of Bonn, Germany} \and Daniel R. Schmidt\thanks{Institute of Computer Science,  University of Cologne, Germany} \and Melanie Schmidt\footnotemark[3]}
\date{}

\begin{document}
\pagenumbering{Alph}
\begin{titlepage}
\maketitle
\thispagestyle{empty}

\begin{abstract}
 Clustering is a fundamental tool in data mining.
 It partitions points into groups (clusters) and may be used to make decisions for each point based on its group.
However, this process may harm protected (minority) classes if the clustering algorithm does not adequately represent them in  desirable clusters -- especially if the data is already biased.

At NIPS 2017, Chierichetti et al.~\cite{CKLV17} proposed a model for \emph{fair clustering} requiring the representation in each cluster to (approximately) preserve the global fraction of each protected class. 
Restricting to two protected classes, they developed both a $4$-approximation for the fair $k$-center problem and a $\mathcal{O}(t)$-approximation for the fair $k$-median problem, where $t$ is a parameter for the fairness model. 
For multiple protected classes, the best known result is a $14$-approximation for fair $k$-center~\cite{RS18}. 

We extend and improve the known results. Firstly, we give a $5$-ap\-prox\-i\-ma\-tion for the fair $k$-center problem with multiple protected classes.
Secondly, we propose a relaxed fairness notion under which we can give bicriteria constant-factor approximations for all of the classical clustering objectives $k$-center, $k$-supplier, $k$-median, $k$-means and facility location. The latter approximations are achieved by a framework that takes an arbitrary existing unfair (integral) solution and a fair (fractional) LP solution and combines them into an essentially fair clustering with a weakly supervised rounding scheme. 
In this way, a fair clustering can be established belatedly, in a situation where the centers are already fixed.
\end{abstract}

\end{titlepage}
\pagenumbering{arabic}

\section{Introduction}
Suppose we are to reorganize school assignments in a big city.
Given a long list of children starting school next year and a short list of all available teachers, the goal is to assign the students-to-be to (public) schools such that the maximum distance to the school is small. 
The school capacity is given by the number of its teachers: For each teacher, $s$ students can be admitted.

This challenge is in fact an instance of the capacitated (metric) $k$-center problem.
A na\"ive solution may, however, result in some schools having an excess of boys while others might have a surplus of girls.
We would prefer an assignment where the classes are more balanced.
Thus a new challenge arises: Assign the children such that the ratio is (approximately) 1:1 between boys and girls, and minimize the maximum distance under this condition.\footnote{Or, incorporating the capacities, ensure that the teacher:boys:girls ratio is $1$:$\frac{s}{2}$:$\frac{s}{2}$.}\enlargethispage{\baselineskip}
This can be modeled by the following combinatorial optimization problem: Given a point set, half of the points are red, the other half is blue. Compute a clustering where each cluster has an equal number of red and blue points, and minimize the maximum radius. 

In this form, our example is a special case of the \emph{fair $k$-center} problem, as proposed by Chierichetti et al.~\cite{CKLV17} in the context of maintaining fairness in \emph{unsupervised} machine learning tasks. Their model is based on the concept of \emph{disparate impact}~\cite{RR14} (and the p\%-rule). The input points are assumed to have a binary sensitive attribute modeled by two colors, and discrimination based on this attribute is to be avoided. Since preserving exact balance in each cluster may be very costly or even be impossible\footnote{Imagine a point set with $49$ red and $51$ blue points: This can not at all be divided into true subsets with exact the same ratio.}, the idea is to ensure that at least $1/t$ of the points of each cluster are of the minority color, where $t$ is a parameter. A cluster with this property is called \emph{fair}, and the fairness constraint can now be added to any clustering problem, giving rise to fair $k$-center, fair $k$-median, etc.
Chierichetti et al. develop a $4$-approximation for fair $k$-center and a $(t+1+\sqrt{3}+\epsilon)$-approximation for fair $k$-median.

The fair clustering model as proposed by Chierichetti et al. can also be used to incorporate other aspects into our school assignment example: For example, we might want to mitigate effects of gentrification or segregation. For these use cases, we need multiple colors. 
Then, in each cluster, the ratio between the number of points with one specific color and the total number of points shall be in some given range. %Given a range for the ratio for each color, we define a constraint that requires the ratio for any color classes inside any given cluster to fall within that range. 
If the allowed range is $[0.20,0.25]$ for red points, we require that in each cluster, at least a fifth and at most a fourth of the points are red. 
This models well established notions of fairness (statistical parity, group fairness), which require that each cluster exhibits the same compositional makeup as the overall data with respect to a given attribute.
One downside of this notion is that a malicious user could create an illusion of fairness by including proxy points: If we wanted to create an boy-heavy school in our above example, we could still achieve the desired parity by assigning only girls that are very unlikely to attend.
Thus, instead of enforcing \emph{equal representation} in the above sense, one could also ask for \emph{equal opportunity} as proposed by Hardt et al.~\cite{HPS16} for the case where we take binary decisions (i.e., $k=2$) and have access to a labeled training set.
This approach, however, raises the philosophical question if this equality of opportunity is a sufficient condition for the absense of discrimination.
Rather than delving into this complex and much debated issue in this algorithmic paper, we refer to the excellent surveys by Romei and Ruggieri~\cite{RR14} and \u{Z}liobait\.{e} et al.~\cite{ZKC11} that systematically discuss different forms of discrimination and how they can be detected.
We assume that it is the intent of the user to achieve a truly fair solution.
  
Finding fair clusterings turns out to be an interesting challenge from the point of view of combinatorial optimization. As other clustering problems with side constraints, it loses the property that points can be assigned locally. But while many other constraint problems at least allow polynomial algorithms that assign points to given centers optimally, we show that even this restricted problem is NP-hard in the case of fair $k$-center.

Chierichetti et al. tackle fair clustering problems by a two-step procedure: First, they compute a micro clustering into so-called \emph{fairlets}, which are groups of points that are fair and can not be split further into true subsets that are also fair. Secondly, representative points of the fairlets are clustered by an approximation algorithm for the unconstrained problem. Consider the special case of a point set with 1:1 ratio of red and blue points. Then a fairlet is a pair of one red and one blue point, and a good micro clustering can be found by computing a suitable bipartite matching between the two color classes. 

The problem of computing good fairlets gets increasingly difficult when considering more general variants of the problem. For multiple colors and the special case of exact ratio preservation (i.e., for all colors, the allowed range for its ratio is one specific number), the fairlet computation problem can be reduced to a capacitated clustering problem. This is used in~\cite{RS18} to obtain a $14/15$-approximation for fair $k$-center/$k$-supplier with multiple colors and exact ratio preservation. 
 
In Appendix~\ref{sec:furtherrelatedwork} and Appendix~\ref{sec:fairletapproach}, we give an extensive overview of the existing results and further the fairlet approach in order to explore its applicability for different variants of fair clustering. Two major issues arise: Firstly, capacitated clustering is not solved for all clustering objectives; indeed, finding a constant-factor approximation for k-median is a long-standing open problem. 
Secondly, (even for $k$-center) it is unclear how fairlets even look like when we have multiple colors and want to allow ranges for the ratios. In this situation, subsets of very different size and composition may satisfy the desired ratio.

The main contribution of this paper is a very different approach. We start with a solution to the unconstrained problem. %Since this can be any solution, we say our algorithm is a \emph{black-box} approximation. 
Based on the given solution, we derive a fair clustering solution with the same centers.
That is achieved by a technique that we call \emph{weakly supervised LP rounding}: We solve an LP for the fair clustering problem and then combine it with the integral unfair solution by careful rounding. We use this method to prove the following statements.

\begin{theorem}\label{thm:exactapprox}
\label{thm:5_approximation:both}
There exists a $5/7$-approximation for the fair $k$-center/$k$-supplier problem with exact preservation of ratios.
\end{theorem}
\begin{theorem}\label{thm:mainresult}
Given any set of centers $S$, there exists an an assignment $\phi':$ which is essentially fair, and which incurrs a cost that is linear in the cost of $S$ for the unconstrained problem and the cost of an optimal fractional fair clustering of $P$, for all objectives $k$-center, $k$-supplier, $k$-median, $k$-means and facility location.
\end{theorem}
\begin{corollary}\label{cor:blackbox}
There exists an essentially fair $3$/$5$/$3.488$/$4.675$/$62.856$-approximation for the fair $k$-center/$k$-supplier/facility location/$k$-median/$k$-means problem.
\end{corollary}

Here, \emph{essentially fair} refers to our notion of bicriteria approximation: A cluster $C$ is \emph{essentially fair} if there exists a fractional fair cluster $C'$, such that for each color $h$ the number of color $h$ points in $C$ differ by \emph{at most $1$} from the mass of color $h$ points in $C'$. So this is a small additive fairness violation. 
For large clusters, an additive violation of $1$ will translate into a negligible multiplicative violation. Of course if a cluster is very small or a color has a very small ratio, then a violation of $1$ is large in multiplicative terms. However, consider an example with one majority color having $500$ points and $50$ colors having $10$ point each. Then our model means that in a cluster of $100$ points, at most one of the minority colors can vanish completely because the \emph{overall} amount of violation is at most one point. So in order to be fair, the majority color can still at most have $51$ points, and at least $49$ of the minority colors are preserved.

For the bicriteria approximations in Theorem~\ref{thm:mainresult} and Corollary~\ref{cor:blackbox}, we can start with \emph{any} unconstrained starting solution. We thus say that our algorithm is a \emph{black-box} approximation. 
For Theorem~\ref{thm:exactapprox}, this is not the case; here we need to use a specific starting solution. We prove Theorem~\ref{thm:mainresult} in Section~\ref{sec:additive_error} and the proof of Theorem~\ref{thm:exactapprox} in Section~\ref{sec:trueapproximations}.

Our results have two advantages. 
Firstly, we get results for a wide range of clustering problems, and these results improve previous results. For example, we get a $5$-approximation for the fair $k$-center problem with exact ratio preservation, where the best known guarantee was $14$. 
All our bicriteria results work for multiple colors and approximate ratio preservation, a case for which no previous algorithm was known. As for the quality of the guarantees, compare the $4.675$-approximation for essentially fair $k$-median clusterings with the best previously known $\Theta(t)$-approximation, which is only applicable to the case of two colors. 
Notice that a similar result can \emph{not} be achieved by using bicriteria approximation algorithms for capacitated clustering. The reduction from capacitated clustering only works when the capacities are not violated.

Secondly, the black-box approach has the advantage that fairness can be established belatedly, in a situation where the centers are already given. \cite{DHPRZ12,ZWSPD13}. 
Consider our school example and notice that the location of the schools cannot be chosen. 
Our result says that if we are alright with essentially fair clusterings, we get a clustering which is not much more expensive than a fair clustering where the centers were chosen with the fairness constraint at hand. 

\paragraph*{Additional related work}
The unconstrained $k$-center/$k$-supplier problem can be $2/3$-approximated~\cite{G85,HS86}, and this is tight~\cite{HN79}.
Facility location can be $1.488$-approximated~\cite{L13}, and the best lower bound is $1.463$~\cite{GK99}. For $k$-median, the best upper bound is $2.765$~\cite{LS16,BPRST17}, while the best hardness result lies  below two~\cite{JMS02}.
A recent breakthrough gives a $6.357$-approximation for $k$-means~\cite{ANSW17}, but the newest hardness result is marginally above 1~\cite{ACKS15,LSW17}.

The $k$-center problem allows for constant-factor approximations for many useful constraints such as capacity constraints~\cite{BKP93,CHK12,KS00}, lower bounds on the size of each cluster~\cite{APFTKKZ10,AS16} or allowing for outliers~\cite{CKMN01,CK14}. 
This is also true for facility location and capacities~\cite{ALBGGGJ13,ASS17,BGG12}, uniform lower bounds~\cite{AS12,S10}, and outliers~\cite{CKMN01}.
Much less is known for $k$-median and $k$-means. True constant-factor approximations so far exist only for the outlier constraint~\cite{C08,KLS18}. A major problem for obtaining constant factor approximations is that the natural LP has an unbounded integrality gap, which is also true for the LP with fairness constraints. 
Bicriteria approximations are known that either violate the capacity constraints~\cite{Li14,Li16,Li17} or the cardinality constraint~\cite{ABGL15}.

A clustering problem where the points have a color was considered by Li, Yi and Zhang~\cite{LYZ10}. They provided a $2$-approximation for a constraint called \emph{diversity}, which allows at most one point per color in each cluster.
Fair clustering as considered in this paper is studied in~\cite{CKLV17} and \cite{RS18}.

\subsection*{Preliminaries}\label{sec:preliminaries}

\paragraph*{Points and locations.}
We are given a set of $n$ points $P$ and a set of potential locations $L$.
We allow $L$ to be infinite (when $L=\R^d$).
The task is to open a subset $S \subseteq L$ of the locations and to assign each point in $P$ to an open location via a mapping $\phi: P \to S$. We refer to the set of all points assigned to a location $i \in S$ by $P(i) := \phi^{-1}(i)$.
The assignment incurs a cost governed by a semi-metric 
$d: (P \cup L) \times (P \cup L) \to \R_{\geq 0}$ %on $P \cup L$
 that fulfills a \emph{$\beta$-relaxed triangle inequality}
\begin{equation}\label{eq:relaxedtriangle} 
d(x,z) \leq \beta(d(x,y) + d(y,z)) \quad \text{for all } x,y,z \in P \cup L 
\end{equation}
for some $\beta \geq 1$. 
Additionally, we may have opening costs $f_i \geq 0$ for every potential location $i \in L$ or 
a maximum number of centers $k\in \N$. 

\paragraph*{Colors and fairness.} We are also given a set of \emph{colors} $Col := \set{col_1,\dots,col_g}$, and a coloring $col: P \to Col$ that assigns a color to each point $j \in P$. For any set of points $P' \subseteq P$ and any color $col_h \in Col$ we define $col_h(P') = \setc{j \in P'}{col(j) = col_h}$ to be the set of points colored with $col_h$ in $P'$. 
We call $r_h(P') := \frac{|col_h(P')|}{|P'|}$ the \emph{ratio} of $col_h$ in $P'$. 
If an implicit assignment $\phi$ is clear from the context, we write $col_h(i)$ to denote the set of all points of a color $col_h \in Col$ assigned to an $i \in S$, i.e., $col_h(i)=col_h(P(i))$.

A set of points $P' \subseteq P$ is \emph{exactly fair} if $P'$ has the same ratio for every color as $P$, i.e., for each $col_h \in Col$ we have $r_h(P') = r_h(P)$.
We say that $P'$ is \emph{$\ell,u$-fair} or just \emph{fair} for some $\ell = (\ell_1 = p^1_1/q^1_1,\ldots, \ell_g = p^g_1/q^g_1)$ and $u = (u_1 = p^1_2/q^1_2,\ldots, u_g = p^g_2/q^g_2)$ if we have $r_h(P') \in [\ell_h,u_h]$ for every color $col_h \in Col$. 

In our fair clustering problems, we want to preserve the ratios of colors found in $P$ in our clusters. We distinguish two cases:
\emph{exact} preservation of ratios, and \emph{relaxed} preservation of ratios.
For the exact preservation of ratios, we ask that all clusters are exactly fair, i.e., $P(i)$ is fair for all $i \in S$.

For the relaxed preservation of ratios, we are given the lower and upper bounds $\ell = (\ell_1 = p^1_1/q^1_1,\ldots, \ell_g = p^g_1/q^g_1)$ and $u = (u_1 = p^1_2/q^1_2,\ldots, u_g = p^g_2/q^g_2)$ on the ratio of colors in each cluster and ask that all clusters are \emph{$\ell,u$-fair}.
The exact case is a special case of the relaxed case where we set $\ell_h = u_h = r_h(P)$ for every color $col_h \in Col$.

\emph{Essentially fair} clusterings are defined below (see Definition~\ref{additive_fairness_violation}).
\paragraph*{Objectives.}
We consider fair versions of several classical clustering problems.
An instance is given by $I := (P,L,col,d,f,k,\ell,u)$, and our goal is to choose a solution $(S,\phi)$ according to one of the following objectives.
\begin{itemize}%[leftmargin=0.4cm,itemsep=0cm,topsep=0.1cm]
	\item \textbf{$k$-center} and \textbf{$k$-supplier}: minimize the maximum distance between a point and its assigned location: $\min \max_{j \in P} d(j,\phi(j))$. In these problems, we have $f \equiv 0$ and $d$ is a metric. Furthermore, in $k$-center, $L = P$, whereas in $k$-supplier , $L \neq P$ is some finite set.
	\item \textbf{$k$-median}: minimize $\sum_{j \in P} d(j,\phi(j))$, $d$ is a metric, $f \equiv 0$ and $L \subseteq P$.
	\item \textbf{$k$-means}: minimize $\sum_{j \in P} d(j,\phi(j))$, where $P \subseteq \R^m$ for some $m \in \N$, $L = \R^m$ and $d(x,y) = ||y-x||^2$ is a semi-metric for $\beta=2$ and $f \equiv 0$.
	\item \textbf{facility location}: minimize $\sum_{j \in P} d(j,\phi(j)) + \sum_{i\in S} f_i$, where $k = n$, $d$ is a metric and $L$ is a finite set. 
\end{itemize}

\paragraph*{The fair assignment problem.} For all the objectives above, we call the subproblem of computing a cost-minimal fair assignment of points to given centers the \emph{fair assignment problem}.
We show the following theorem in the full version.
\begin{theorem}
Finding an $\alpha$-approximation for the fair assignment problem for $k$-center for $\alpha < 3$ is NP-hard.
\end{theorem}

\subsection*{(I)LP formulations for fair clustering problems}
\label{sec:fairassignment}
\newcommand{\inlineeqnum}{\refstepcounter{equation}~\mbox{(\theequation)}~}
Let $I = (P,L,col,d,f,k,\ell,u)$ be a problem instance for a fair clustering problem.
We introduce a binary variable $y_i \in \set{0,1}$ for all $i \in L$ that decides if $i$ is opened, \ie $y_i = 1 \Leftrightarrow i \in S$. Similarly, we introduce binary variables $x_{ij} \in \set{0,1}$ for all $i \in L, j \in P$ with $x_{ij} = 1$ if $j$ is assigned to $i$, \ie $\phi(j) = i$. 
All ILP formulations have the inequalities $\inlineeqnum\label{eq:general:lppointassigned}\sum\limits _{i \in L} x_{ij} = 1\ \forall j \in P$ saying that every point $j$ is assigned to a center, the inequalities $\inlineeqnum\label{eq:general:opening} x_{ij} \leq  y_i\ \forall i\in L, j \in P$ ensuring that if we assign $j$ to $i$, then $i$ must be open, and the integrality constraints $\inlineeqnum\label{eq:general:integral} y_i,x_{ij}  \in \{0,1\}\ \forall i\in L, j\in P$.
We may restrict the number of open centers to $k$ with
$\inlineeqnum\label{eq:general:klocations} \sum_{i \in L} y_i \leq k$.
For $k$-center and $k$-supplier, the objective is commonly encoded in the constraints of the problem, and the (I)LP has no objective function.
The idea is to guess the optimum value $\tau$. Since there is only a polynomial number of choices for $\tau$, this is easily done. Given $\tau$, we construct a \emph{threshold graph} $G_\tau=(P\cup L,E_\tau)$ on the points and locations, where a connection between $i \in L$ and $j \in P$ is added iff $i$ and $j$ are close, i.e., $\{i,j\}\in E_{\tau} \Leftrightarrow d(i,j) \le \tau$. Then, the we ensure that points are not assigned to centers outside their range:
\begin{align}
\label{eq:obj:kcenter} &&                             x_{ij} &= 0  &&\quad \text{for all } i\in L, j\in P, \{i,j\} \notin E_{\tau}
\end{align}
For the remaining clustering problems, we pick the adequate objective function from the following three
(let $d_{ij} := \dist(i,j)$):

\noindent\begin{minipage}{0.28\textwidth}
\begin{equation}
\label{eq:obj:kmedian}\min \slrsum{i\in L, j \in P} x_{ij} d_{ij}
\end{equation}
\end{minipage}
\begin{minipage}{0.32\textwidth}
\begin{equation}
\label{eq:obj:kmeans} \min \slrsum{i\in L, j \in P} x_{ij} d_{ij}^2
\end{equation}
\end{minipage}%
\begin{minipage}{0.39\textwidth}
\begin{equation}
\label{eq:obj:fl} \min \slrsum{i\in L, j \in P} x_{ij} d_{ij} + \sum_{i\in L} y_i f_i
\end{equation}
\end{minipage}

We now have all necessary constraints and objectives. For $k$-center and $k$-supplier, we use inequalities~\eqref{eq:general:lppointassigned}-\eqref{eq:obj:kcenter}, no objective, and define the optimum to be the smallest $\tau$ for which the ILP has a solution. We get $k$-median and $k$-means by combining inequalities~\eqref{eq:general:lppointassigned}-\eqref{eq:general:klocations} with~\eqref{eq:obj:kmedian} and~\eqref{eq:obj:kmeans}, respectively, and we get facility location by combining \eqref{eq:general:lppointassigned}-\eqref{eq:general:integral} with the objective~\eqref{eq:obj:fl}. LP relaxations arise from all ILP formulations by replacing~\eqref{eq:general:integral} by $y_i,x_{ij}  \in [0,1]$ for all $i\in L, j\in P$.
To create the fair variants of the ILP formulations, we add fairness constraints modeling the upper and lower bound on the balances.
\begin{align}
\label{eq:general:fair} \ell_h \slrsum{j \in P} x_{ij} \leq \srsum{col(p_j) = col_h} x_{ij} \leq u_h \slrsum{j \in P} x_{ij}\ &&\text{for all}\ i\in L, h \in Col 
\end{align}

Although very similar to the canonical clustering LPs, the resulting LPs become much harder to round even for $k$-center with two colors. 
We show the following in Section~\ref{sec:appendix-integrality-gap}.
\begin{restatable}{lemma}{lemintgap}\label{lem:lp-integrality-gap}
  There is a choice of non-trivial fairness intervals such that the integrality gap of the LP-relaxation of the canonical fair clustering ILP 
	is~$\Omega(n)$ for 
	the fair $k$-center/$k$-supplier/facility location problem.
	The integrality gap is $\Omega(n^2)$ for the fair $k$-means problem.
\end{restatable}

\paragraph*{Essential fairness.}
For a point set $P'$, $\mass_h(P') = |col_h(P')|$ is the \emph{mass} of color $col_h$ in $P'$.
For a possibly fractional LP solution $(x,y)$, we extend this notion to $\mass_h(x,i) := \sum_{j \in col_h(P)} x_{ij}$. We denote the total mass assigned to $i$ in $(x,y)$ by $\mass(x,i) = \sum_{j \in P} x_{ij}$. % and the \emph{ratio of a color $col_h \in Col$ assigned to a center $i \in L$} is then $r_h(i) := \mass_h(x,i) / \mass(x,i)$.
With this notation, we can now formalize our notion of \emph{essential fairness}.

\begin{definition}[Essential fairness]\label{additive_fairness_violation}
 Let $I$ be an instance of a fair clustering problem and let $(x,y)$ be an integral, but not necessarily fair solution to $I$.
 We say that $(x,y)$ is \emph{essentially fair} if there exists a fractional fair solution $(x',y')$ for $I$ 
 such that $\forall i \in L$:
 \begin{align}
  &\lfloor\mass_h(x',i)\rfloor  
	\leq \mass_h(x,i)
	\leq \lceil\mass_h(x',i)\rceil \qquad &\forall col_h \in Col\\
  \text{and\quad } &\lfloor\mass(x',i)\rfloor  
	\leq \mass(x,i)
	\leq \lceil\mass(x',i)\rceil.&
 \end{align}
\end{definition}

\section{Black-box approximation algorithm for essentially fair clustering}
\label{sec:additive_error}

For essentially fair clustering, we give a powerful framework that employs approximation algorithms for (unfair) clustering problems as a black-box and transforms their output into an essentially fair solution.
In this framework, we start by computing an approximate solution for the standard variant of the 
clustering problem at hand. Next, we solve the LP for the fair variant of
the clustering problem. Now we have an integral unfair solution, and a
fractional fair solution. Our final and most important step is to
combine these two solutions into an integral and essentially fair solution.
It consists of two conceptual sub-steps: Firstly, we show that it is possible
to find a fractional fair assignment to the centers of the integral solution that is sufficiently cheap. Secondly, we round the assignment. This last sub-step introduces
the potential fairness violation of one point per color per cluster. Algorithms~\ref{alg:blackbox} and \ref{alg:blackboxAlpha} summarize this framework.

\begin{algorithm}[t]
 \SetAlgoLined
 \DontPrintSemicolon
 \KwData{An instance $I$ of a clustering problem, and an integral solution $(\bar{x},\bar{y})$ to $I$.}
 \KwResult{An essentially fair solution to $I$ with a cost as specified in Theorem~\ref{thm:blackboxRaw}.}
 \BlankLine
 \nl Compute an optimal solution to the LP-relaxation $(x^{LP},y^{LP})$ of $I$ (see Section~\ref{sec:fairassignment}).\;
 \nl Combine $(\bar{x},\bar{y})$ and $(x^{LP},y^{LP})$ using Lemmas~\ref{lemma:yrounding} and \ref{lemma:xrounding} into an essentially fair solution to $I$.
 \caption{Black-box assignment with essentially fair solutions.\label{alg:blackbox}}
\end{algorithm}

\begin{algorithm}[t]
 \SetAlgoLined
 \DontPrintSemicolon
 \KwData{An instance $I$ of a clustering problem.}
 \KwResult{An essentially fair solution to $I$ with an approximation factor as specified in Theorem~\ref{thm:blackbox}.}
 \BlankLine
 \nl Compute an $\alpha$-approximate, integral solution $(\bar{x},\bar{y})$ to $I$.\;
 \nl Use Algorithm~\ref{alg:blackbox} with $I$ and $(\bar{x},\bar{y})$ as input to obtain an essentially fair solution to $I$ with an approximation factor as specified by Theorem~\ref{thm:blackbox}.\;
 \caption{Black-box approximation algorithm for essentially fair clustering.\label{alg:blackboxAlpha}}
\end{algorithm}

We show that this approach yields constant-factor approximations with fairness violation for all mentioned clustering objectives. The description will be neutral whenever the objective does not matter. Thus, descriptions like \emph{the LP} mean the appropriate LP for the desired clustering problem. When the problem gets relevant, we will specifically discuss the distinctions. Notice that for all clustering problems defined in Section~\ref{sec:preliminaries}, $P$ and $L$ are finite except for $k$-means. However, for the $k$-means problem, we can assume that $L=P$ if we accept an additional factor of $2$ in the approximation guarantee (Lemma~\ref{lemma:twofactorkmeans} in the appendix). Thus, we assume in the following that $L$ and $P$ are finite sets. Indeed, we even assume at least $L\subseteq P$ for all problems except $k$-supplier and facility location.

\newcommand{\lpx}{x^{LP}\xspace}
\newcommand{\lpy}{y^{LP}\xspace}
\newcommand{\lpcost}{c^{LP}\xspace}
\newcommand{\optcost}{c^\ast\xspace}
\newcommand{\approxx}{\bar x\xspace}
\newcommand{\approxy}{\bar y\xspace}
\newcommand{\approxcost}{\bar c\xspace}
\newcommand{\approxassignment}{\bar \phi\xspace}
\newcommand{\approxcluster}{\bar C\xspace}
\newcommand{\firstx}{\hat x\xspace}
\newcommand{\firsty}{\hat y\xspace}
\newcommand{\firstcost}{\hat c\xspace}

\subsection{Step 1: Combining a fair LP solution with an integral unfair solution}
In the first step, we assume that we are given two solutions. Let $(\lpx,\lpy)$ be an optimal solution to the LP. This solution has the property that the assignments to all centers are fair, however, the centers may be fractionally open and the points may be fractionally assigned to several centers. 
Let $\lpcost$ be the objective value of this solution. For $k$-supplier and $k$-center, it is the smallest $\tau$ for which the LP is feasible, for the other objectives, it is the value of the LP. 
We denote the cost of the best \emph{integral} solution to the LP by $\optcost$. We know that $\lpcost \le \optcost$.

Let $(\approxx,\approxy)$ be any integral solution to the LP that may violate fairness, i.e., inequality~\eqref{eq:general:fair}, and let $\approxcost$ be the objective value of this solution. 
We think of $(\approxx,\approxy)$ as being a solution of an $\alpha$-approximation algorithm for the standard (unfair) clustering problem for some constant $\alpha$. Since the unconstrained version can only have a lower optimum cost, we then have $\approxcost \le \alpha \cdot \optcost$.

Our goal is now to combine $(\lpx,\lpy)$ and $(\approxx, \approxy)$ into a third solution, $(\firstx, \firsty)$, such that the cost of $(\firstx, \firsty)$ is bounded by $O(\lpcost + \approxcost) \subseteq O(\optcost)$. Furthermore, the entries of $\firsty$ shall be integral. The entries of $\firstx$ may still be fractional after step 1.

Let $S$ be the set of centers that are open in $(\approxx,\approxy)$. 
For all $j \in P$, we use $\approxassignment(j)$ to denote the center in $S$ closest to $j$, i.e., $\approxassignment(j) = \arg\min_{i \in S} \dist(j,i)$ (ties broken arbitrarily). Notice that the objective value of using $S$ with assignment $\approxassignment$ for all points in $P$ is at most $\approxcost$, since assigning to the closest center is always optimal for the standard clustering problems without fairness constraint.

Depending on the objective, $L$ is a subset of $P$ or not, i.e., $\approxassignment$ is not necessarily defined for all locations in $L$. We then extend $\approxassignment$ in the following way. Let $i\in L\backslash P$ be any center, and let $j^\ast$ be the closest point to it in $P$. Then we set $\approxassignment(i) := \approxassignment(j^\ast)$, i.e., $i$ is assigned to the center in $S$ which is closest to the point in $P$ which is closest to $i$.
Finally, let $\approxcluster(i) = \approxassignment^{-1}(i)$ be the set of all points and centers assigned to $i$ by $\approxassignment$. 

\begin{lemma}\label{lemma:yrounding}
Let $(\lpx,\lpy)$ and $(\approxx, \approxy)$ be two solutions to the LP, where $(\approxx, \approxy)$may violate inequality~\eqref{eq:general:fair}, but is integral. Then the solution defined by
\begin{align*}
  \firsty := \approxy, \qquad \firstx_{ij} := \sum_{i' \in \approxcluster(i)} \lpx_{i'j} \quad\text{for all } i \in S, j \in P, \qquad
\firstx_{ij} := 0 \quad \text{for all }i \notin S, j \in P.
\end{align*}
satisfies inequality~\eqref{eq:general:fair}, $\firsty$ is integral, and the cost $\firstcost$ of $(\firstx,\firsty)$ is bounded by $\lpcost+\approxcost$ for $k$-center, by $2\cdot \lpcost + \approxcost$ for $k$-supplier, $k$-median, and facility location, and by $12 \cdot \lpcost + 8 \cdot \approxcost$ for $k$-means.
\end{lemma}

\begin{proof}
Recall that for $k$-center and $k$-supplier, speaking of the cost of an LP solution is a bit sloppy; we mean that $(\firstx,\firsty)$ is a feasible solution in the LP with threshold $\firstcost$.

The definition of $(\firstx,\firsty)$ means the following. For every (fractional) assignment from a point $j$ to a center $i'$, we look at the cluster with center $i = \approxassignment(i')$ to which $i'$ is assigned to by $\approxassignment$. We then shift this assignment to $i$. So from the perspective of $i$, we collect all fractional assignments to centers in $\approxcluster(i)$ and consolidate them at $i$. Notice that the (fractional) number of points assigned to $i$ after this process may be less than one since $(\approxx,\approxy)$ may include centers that are very close together.

Since that $\firsty$ is simply $\approxy$ it is integral as well and has the same number of centers, thus $\firsty$ also satisfies~\eqref{eq:general:klocations} if the problem uses it. 
Next, we observe that $(\firstx,\firsty)$ satisfies fairness, i.e., respects inequality~\eqref{eq:general:fair}. This is true because $(\lpx,\lpy)$ satisfies them, and because we move \emph{all} assignment from a center $i'$ to the same center $\approxassignment(i')$. This shifting operation preserves the fairness.
Inequality~\eqref{eq:general:opening} is true because we only move assignments to centers that are fully open in $(\approxx,\approxy)$, i.e., the inequality can not be violated as long as~\eqref{eq:general:lppointassigned} is true (which it is for $(\lpx,\lpy)$ since it is a feasible LP solution). Equality~\eqref{eq:general:lppointassigned} is true for $(\firstx,\firsty)$ since all assignment of $j$ is moved to some fully open center. Thus $(\firstx,\firsty)$ is a feasible solution for the LP.
It remains to show that $\firstcost$ is small enough, which depends on the objective.\\[2mm]
\textbf{$k$-median and $k$-means.} We start by showing this for $k$-median (where the distances are a metric, i.e., $\beta=1$ in the $\beta$-triangle inequality \eqref{eq:relaxedtriangle}) and $k$-means (where the distances are a semi-metric with $\beta=2$).
We observe that here, the cost of $(\firstx,\firsty)$ is
\[
\firstcost=\sum_{j \in P}\sum_{i \in L} \firstx_{ij} \dist(i,j)
= \sum_{j\in P}\sum_{i \in L}\sum_{i' \in \approxcluster(i)} \lpx_{i'j} \dist(i,j).
\]
 Now fix $i\in L$, $i' \in \approxcluster(i)$ and $j\in P$ arbitrarily. 
By the $\beta$-relaxed triangle inequality, $\dist(i,j) \le \beta\cdot\dist(i',j) + \beta\cdot\dist(i',i)$. 
Furthermore, we know that $i' \in \approxcluster(i)$, i.e., $\approxassignment(i') = i$ and $\dist(i',i) \le \dist(i',\approxassignment(j))$. We can use this to relate $\dist(i',i)$ to the cost that $j$ pays in $(\approxx,\approxy)$:
\[
\dist(i',i) \le \dist(i',\approxassignment(j)) \le \beta\cdot\dist(j,i') + \beta\cdot\dist(j,\approxassignment(j)).
\]
Adding this up yields
\begin{align*}
&& \sum_{j\in P}\sum_{i \in L}\sum_{i' \in \approxcluster(i)} \lpx_{i'j} \dist(i,j)\\
&\le&
\sum_{j\in P}\sum_{i \in L}\sum_{i' \in \approxcluster(i)} (\beta+\beta^2) \lpx_{i'j} \dist(i',j)
&+ \sum_{j\in P}\sum_{i \in L}\sum_{i' \in \approxcluster(i)} \beta^2\cdot\lpx_{i'j} \dist(j,\approxassignment(j))\\
& =& (\beta+\beta^2)\cdot \lpcost &+ \beta^2\cdot\approxcost.
\end{align*}
For $\beta=1$ ($k$-median), this is $2 \lpcost+\approxcost$, for $\beta=2$ ($k$-means), we get $12\lpcost + 8 \approxcost$ due to the additional factor of $2$ from Lemma\ref{lemma:twofactorkmeans}.\\[2mm]
\textbf{Facility location.} For facility location, we have to include the facility opening costs. We open the facilities that are open in $(\approxx,\approxy)$, which incurs a cost of $\sum_{i \in L} \approxy_i f_i$. The distance costs are the same as for $k$-median, so we get a total cost of
\begin{align*}
\sum_{j\in P}\sum_{i \in L}\sum_{i' \in \approxcluster(i)} 2 \lpx_{i'j} \dist(i',j)
+ \sum_{j\in P}\sum_{i \in L}\sum_{i' \in \approxcluster(i)} \lpx_{i'j} \dist(j,\approxassignment(j))+\sum_{i \in L} \approxy_i f_i
\le\  2 \lpcost + \approxcost.
\end{align*}

\noindent\textbf{$k$-center and $k$-supplier.} For the $k$-center and $k$-supplier proof, we again fix $i\in L$, $i' \in \approxcluster(i)$ and $j\in P$ arbitrarily and use that $\dist(i,j) \le \dist(i,i') + \dist(i',j)$. Now for $k$-center, we know that $\dist(i,i') \le \approxcost$ since $i' \in \approxcluster(i)$, and we know that $\dist(i',j) \le \lpcost$ for all $j$ where $\lpx_{ij}$ is strictly positive. Thus, if $\firstx_{ij}$ is strictly positive, then $\dist(i,j) \le \approxcost + \lpcost$.
For $k$-supplier, we have no guarantee that $\dist(i,i') \le \approxcost$ since $i'$ is not necessarily an input point. Instead, $i' \in \approxcluster(i)$ means that the point $j'$ in $P$ which is closest to $i'$ is assigned to $i$ by $\approxx$. Since $j'$ is the closest to $i'$ in $P$, we have $\dist(i',j')\le \dist(i',j)$. Furthermore, since $j' \in \approxcluster(i)$, $\dist(i,j')\le\approxcost$. Thus, we get for $k$-supplier that
\begin{align*}
\dist(i,j) \le \dist(i,i') + \dist(i',j) \le \dist(i,j') + \dist(i',j')+\dist(i',j)
\le \approxcost + 2 \cdot \lpcost. 
\end{align*}
\end{proof}

\subsection{Step 2: Rounding the \texorpdfstring{$x$}{x}-variables}

For rounding the $x$-variables, we need to distinguish between two cases of objectives. Let $j \in P$ be a point that is fractionally assigned to some centers $L_j \subseteq L$. 

First, we have objectives where we can shift mass from an assignment of $j$ to $i' \in L_j$ to an assignment of $j$ to $i'' \in L_j$ without modifying the objective. We say that such objectives are \emph{reassignable} (in the sense that we can reassign $j$ to centers in $L_j$ without changing the cost). $k$-center and $k$-supplier have this property.

Second, we have objectives where the assignment cost is separable, i.e., where the distances influence the cost via a term of the form $\sum_{i \in L, j \in P} c_{ij} \cdot x_{ij}$ for some $c_{ij} \in \mathbb{R}_{\geq 0}$. 
We call such objectives \emph{separable}. Facility location, $k$-median and $k$-means fall into the this category. 

\begin{lemma}\label{lemma:xrounding}
 Let $(x,y)$ be an $\alpha$-approximative fractional solution for a fair clustering problem with the property that all $y_i, i \in L$ are integral. Then we can obtain an $\alpha$-approximative integral solution $(x',y')$ with an additive fairness violation of at most one in time $O(\poly(|S|+|P|))$, with $S := \setc{i \in L}{y_i \geq 1}$ being the set of locations that are opened in $(x,y)$.
\end{lemma}

\begin{proof}
 We create our rounded $\alpha$-approximate integral solution $(x',y')$ by min-cost flow computations. We begin by constructing a min-cost flow instance which depends on our starting solution $(x,y)$ as well as on the objective of the problem we are studying.

We define a min-cost flow instance $(G = (V,A), c, b)$ (also see Figure~\ref{figure:flow}) with unit capacities and costs $c$ on the edges as well as balances $b$ on the nodes. We begin by defining a graph $G^h = (V^h,A^h)$ for every color $h \in Col$ with
 \[\begin{aligned}
  V^h &:= V_S^h \cup V_P^h, \quad V_S^h := \setc{v^h_i}{i \in S}, \quad V_P^h := \setc{v^h_j}{j \in col_h(P)},\\
	A^h &:= \setc{(v^h_j,v^h_i)}{i \in S, j \in col_h(P):x_{ij} > 0},
 \end{aligned}\]
 as well as costs $c^h$ by
 $c^h_a := c_{ij}$ for $a = (v^h_j,v^h_i) \in A^h, i \in S, j \in col_h(P)$ 
 and balances $b^h$ by
 $ b^h_v := 1$ if $v \in V_P^h$ and $b_v^h := -\lfloor\mass_h(x,i)\rfloor$ if $v = v^h_i \in V_S^h$.
We use the graphs $G_h$ to define $G = (V,A)$ by
 \[\begin{aligned}
  V := &\{t\} \cup V_S \cup \bigcup_{h \in Col} V^h, \quad V_S := \setc{v_i}{i \in S}\\
	A := &\bigcup_{h \in Col} A^h \cup \setc{(v^h_i,v_i)}{i \in S, h \in Col: \mass_h(x,i) - \lfloor\mass_h(x,i)\rfloor > 0}\\
			&\cup \setc{(v_i,t)}{i \in S: \mass(x,i) - \lfloor\mass(x,i)\rfloor > 0} ,
 \end{aligned}\]
 together with costs $c$ of $c_a := c^h_a$ for $a \in A^h$ and $0$ otherwise,
 and balances $b$ of $b_v := b^h_v$ if $v \in V^h$ for some $h \in Col$, $b_v := -B_i$ if $v = v_i \in V_S$ and $b_t := -B$
with $B_i = \lfloor\mass(x,i)\rfloor -\sum_{h \in Col}\lfloor\mass_h(x,i)\rfloor$ and $B := |P| - \sum_{i\in S} \lfloor\mass(x,i)\rfloor$. 

 \paragraph*{Separable objectives -- $k$-median and $k$-means.}
 We make the following observations.
 \begin{enumerate}[leftmargin=0.5cm,itemsep=-0.1cm,topsep=0.1cm]
  \item $B$  and $B_i$ are integers for all $i \in S$, and so are all capacities, costs and balances. Consequently, there are integral optimal solutions for the min-cost flow instance $(G, c, b)$,
	\item $(x,y)$ induces a feasible solution for $(G, c, b)$, by defining a flow $x$ in $G$ as follows:
	 \[ x_a := \begin{cases} x_{ij} & \text{if } a = (v^h_j,v^h_i) \in A^h, j \in P, i \in S,\\ 
	                         \mass_h(x,i) - \lfloor\mass_h(x,i) \rfloor& \text{if } a = (v^h_i,v_i) \in A, h \in Col, i \in S,\\
													\mass(x,i) -  \lfloor\mass(x,i) \rfloor& \text{if } a = (v_i,t) \in A, i \in S.\\
	                            \end{cases} \]	 
   Since $(x,y)$ is a fractional solution, $x$ satisfies capacity and non-negativity constraints because $x_{ij} \in [0,1]$ for all $i \in L, j \in P$ and $\mass_h(x,i) - \lfloor\mass_h(x,i)\rfloor, \mass(x,i) - \lfloor\mass(x,i)\rfloor \in [0,1]$ for all $i \in S$ and $col_h \in Col$ as well. We have flow conservation since the fractional solution needs to assign all points, and the flow of the edges $(v^h_i,v_i)$ and $(v_i,t)$ as well as the demand of $v_i$ and $t$ are chosen in such a way that we have flow conservation for all the other nodes as well.
	\item Integral solutions $x$ to the min-cost flow instance $(G, c, b)$ induce an integral solution $(\bar{x},y)$ to the original clustering problem by setting $\bar{x}_{ij} := x_{a}$ for $a = (v^h_j,v^h_i) \in A^h$ if $j \in col_h(P), i \in S$. Since the flow $x$ is integral, this gives us an integral assignment of all points to centers which have been opened, since $y$ was already integral before this step.
	
	This incurs the additive fairness violation of at most one, since every $i \in S$ is guaranteed by our balances to have at least $\lfloor\mass_h(x,i)\rfloor$ points of color $h \in Col$ and at least $\lfloor\mass(x,i)\rfloor$ points in total assigned to it. Since there is at most one outgoing arc of unit capacity $(v^h_i,v_i)$ and $(v_i,t)$ for an $i \in S$ if $\mass_h(x,i)- \lfloor\mass_h(x,i)\rfloor > 0$, we have at most $\lceil\mass_h(x,i)\rceil$ points of color $col_h$ and $\lceil\mass(x,i)\rceil$ total points assigned to $i$.
 \end{enumerate}	
Together, this yields that computing a min-cost flow $\hat{x}$ for $(G, c, b)$ followed by applying the third observation to $\hat{x}$ yields a solution $(\hat{x},y)$ to the clustering with an additive fairness violation of at most one.

Since $(x,y)$ was inducing the fractional solution $x$ with $\cost(x) = \cost(x,y)$ to the min-cost flow instances, and $\cost(x) \geq \cost(\hat{x})$ by construction we have $\cost(\hat{x},y) \leq \cost(x,y)$.
 
 \paragraph*{Reassignable objectives -- $k$-center and $k$-supplier.}
 In the case of reassignable objectives, we do not have to care about costs, as long as the reassignments happen to centers in $L_j$ for all points $j \in P$. We essentially use the same strategy as before, but instead of a min cost flow problem we solve the transshipment problem $(G = (V,A), b)$ with unit capacities on the edges and balances $b$ on the nodes.
Notice that the three observations from the previous case apply here as well, and reassignability guarantees that the cost does not increase.
\end{proof}
Lemmas~\ref{lemma:yrounding} and \ref{lemma:xrounding} then lead directly to Theorem~\ref{thm:mainresult}, formulated in more detail in the following theorem. 

\begin{theorem}\label{thm:blackboxRaw}
 Black-box approximation for fair clustering gives essentially fair solutions with a cost of $c^{LP}+\bar{c}$ for $k$-center, $2c^{LP}+\bar{c}$ for $k$-supplier, $k$-median and facility location, and $12c^{LP}+8\bar{c}$ for $k$-means where $c^{LP}$ is the cost of an optimal solution to the fair LP relaxation and $\bar{c}$ is the cost of the unfair clustering solution.
\end{theorem}
We know that $c^{LP}$ is not more expensive than an optimal solution to the fair clustering problem. If we use an $\alpha$-approximation to obtain the unfair clustering solution, we have that $\bar{c}$ is at most $\alpha$ times the cost of an optimal solution to the fair clustering problem. Currently, the best know approximation factors are $2$ for $k$-center \cite{G85,HS86}, $3$ for $k$-supplier~\cite{HS86}, $1.488$ for facility location~\cite{L13}, $2.675$ for $k$-median~\cite{BPRST17,LS16} and $6.357$ for $k$-means~\cite{ANSW17}, yielding the following theorem.
\begin{theorem}\label{thm:blackbox}
 Black-box approximation for fair clustering gives essentially fair solutions with an approximation factor of $3$ for $k$-center, $5$ for $k$-supplier, $4.675$ for $k$-median, $3.488$ for facility location, and $62.856$ for $k$-means.
\end{theorem}

%\documentclass{standalone}
%\usepackage[english]{babel}
%\usepackage{tikz}
%\usetikzlibrary{arrows}
%\usetikzlibrary{decorations,arrows}
%\usetikzlibrary{patterns}
%\usepgflibrary{decorations.pathreplacing}
%\usepgflibrary{arrows}
%\usepackage{amsmath,amsthm,amssymb,amsfonts,mathtools,stmaryrd}
%
%\begin{document}
\begin{figure}
\centering
\begin{tikzpicture}[scale=1.2,
		pp/.style={point, black},
    center/.style={draw, gray!80!white, fill=gray!10!white, rounded corners, very thick},
    >=stealth, thick
    ]   

\node[] at (-1.5,3) {\underline{Nodes for:}};
%Points sorted by color
			\node at (-1.5,-1) {$P$};
      \draw (0,-1) node[bp] (u1) {};
      \draw (1.25,-1) node[bp] (u2) {};
      \draw (2.5,-1) node[bp] (u3) {};
      \draw (3.5,-1) node[rp] (u4) {};
      \draw (4.75,-1) node[rp] (u5) {};
      \draw (6,-1) node[rp] (u6) {};
			
%Circle connecting the Nodes per Center				
			\draw[center] (1,1.6) -- (1.92,0.27) -- (0.08,0.27) -- cycle;			
			\draw[center] (3,1.6) -- (3.92,0.27) -- (2.08,0.27) -- cycle;			
			\draw[center] (5,1.6) -- (5.92,0.27) -- (4.08,0.27) -- cycle;
					
			\node at (1,0.7) {$c_1$};
			\node at (3,0.7) {$c_2$};
			\node at (5,0.7) {$c_3$};
			
%Node for each Center and Color	combination
			\node at (-1.5,0.5) {$S,h$};
			\draw (0.5,0.5) node[bp] (f1a1) {};
      \draw (1.5,0.5) node[rp] (f1a2) {};
			\draw (2.5,0.5) node[bp] (f2a1) {};
      \draw (3.5,0.5) node[rp] (f2a2) {};
			\draw (4.5,0.5) node[bp] (f3a1) {};
      \draw (5.5,0.5) node[rp] (f3a2) {};
			
%Node for each Center
			\node at (-1.5,1.2) {$S$};
      \draw (1,1.2) node[pp] (f1) {};
      \draw (3,1.2) node[pp] (f2) {};
      \draw (5,1.2) node[pp] (f3) {};
			
%Sink $t$			
			\node at (-1.5,2.4) {$t$};
			\draw (3,2.4) node[pp] (t) {t};
			
%Connect Points with Center,Color Nodes they are partially assigned to			
			\draw[<-] (f1a1) -- (u1); 
			\draw[<-] (f1a1) -- (u2); 
			\draw[<-] (f3a1) -- (u3); 
			\draw[<-] (f2a1) -- (u3);
			\draw[<-] (f2a1) -- (u2);
			\draw[<-] (f1a2) -- (u4);
			\draw[<-] (f2a2) -- (u4);
			\draw[<-] (f3a2) -- (u5);
			\draw[<-] (f1a2) -- (u5); 
			\draw[<-] (f3a2) -- (u6);
			\draw[<-] (f2a2) -- (u6);

%Connect the Center,Color Nodes with the Center node
			\draw[<-] (f1) -- (f1a1);\draw[<-] (f1) -- (f1a2);
			\draw[<-] (f2) -- (f2a1);\draw[<-] (f2) -- (f2a2);
			\draw[<-] (f3) -- (f3a1);\draw[<-] (f3) -- (f3a2);

% Connect center nodes with $t$
			\draw[<-] (t) -- (f1); 
			\draw[<-] (t) -- (f2); 
			\draw[<-] (t) -- (f3);
			
%b-values
			
			\node[anchor=west] at (6.5,3) {\underline{b-values}};
			\node[anchor=west] at (6.5,2.4) {$-B$};	
			\node[anchor=west] at (6.5,1.2) {$-B_i$};
			\node[anchor=west] at (6.5,0.5) {$-\lfloor \mass_h(x,i) \rfloor$};
			\node[anchor=west] at (6.5,-1) {$1$};
\end{tikzpicture}
\caption{Example for the graph $G$ used in the rounding of the $x$-variables.\newline $B_i = \lfloor \mass(x,i) \rfloor - \sum_{h \in Col} \lfloor \mass_h(x,i)\rfloor$ and $B = |P| - \sum_{i \in S}\lfloor\mass(x,i)\rfloor$.\label{figure:flow}}
\end{figure}
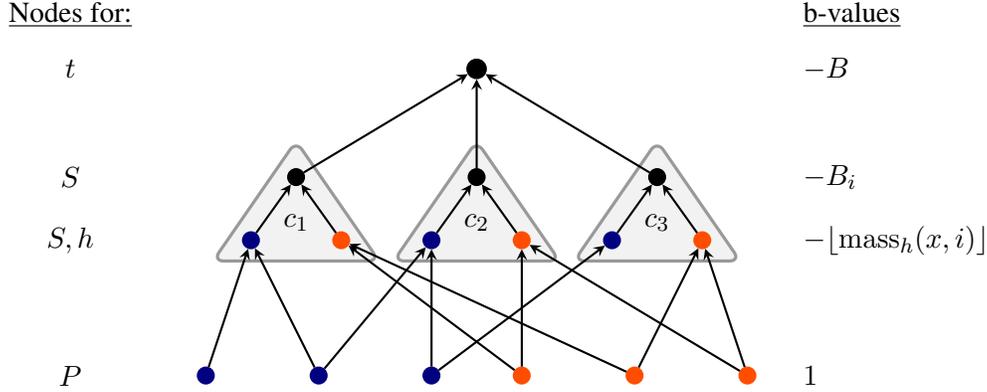

\section{True approximations for fair \texorpdfstring{$k$}{k}-center and  \texorpdfstring{$k$}{k}-supplier}\label{sec:trueapproximations}

We now extend our weakly supervised rounding technique for $k$-center and $k$-supplier in the case of the exact fairness model. We replace the black-box algorithm with a specific approximation algorithm, and then achieve true approximations for the fair clustering problems by informed rounding of the LP solution.

\subsection{\texorpdfstring{$5$}{5}-Approximation Algorithm for \texorpdfstring{$k$}{k}-center}\label{sec:five_approx}
In this section, we consider the fair $k$-center problem with exact preservation of ratios and without any additive fairness violation.

We give a $5$-approximation for this variant.
The algorithm begins by choosing a set of centers.
In contrast to Section~\ref{sec:additive_error} we do not use an arbitrary algorithm for the standard $k$-center problem but specifically look for nodes in the threshold graph $G_{\tau} = (P, E_{\tau})$ where $E_{\tau} = \{(i,j) \mid i\neq j \in P, d(i,j) \leq \tau \}$ that form a maximal independent set $S$ in $G_{\tau}^2$.  Here $G_{\tau}^t$ denotes the graph on $P$ that connects all pairs of nodes with a distance at most $t$ in $G$ and we denote the edge set of $G_{\tau}^t$ by $E_{\tau}^t$.
As we use the following procedure independent for each connected component of $G_\tau$, we will in the description and the following proofs of the procedure assume that $G_\tau$ is a connected graph.
The procedure uses the approach by Khuller and Sussmann ~\cite{KS00} (procedure \textsc{AssignMonarchs}) to find $S$ which ensures the following property: There exists a tree $T$ spanning all the nodes in $S$ and two adjacent nodes in $T$ are exactly distance $3$ apart in $G_\tau$. The procedure begins by choosing an arbitrary vertex $r \in P$, called \emph{root}, into $S$ and marking every node within distance $2$ of $r$ (including itself). Until all the nodes in $P$ are marked, it chooses an unmarked node $u$ that is adjacent to a marked node $v$ and marks all nodes in the distance two neighborhood of $u$. Observe that $u$ is exactly at distance $3$ from a node $u' \in S$ chosen earlier that caused $v$ to get marked. Thus the run of the procedure implicitly defines the tree $T$ over the nodes of $S$.
In case $G_\tau$ is not a connected graph this procedure is run on each connected component and the set $S$ has the following property: There exists a forest $F$ such that $F$ reduced to a connected component of $G_\tau$ is a tree $T$ spanning all the nodes of $S$ inside of that connected component and two adjacent nodes in $T$ are exactly distance $3$ apart in $G_\tau$.

In the next phase, we make use of some structure that feasible solutions with exact preservation of the ratios must have.
\begin{observation}
\label{obs:fairletmultiples}
Let $m \in \mathbb{N}$ be the smallest integer such that for each color $h \in Col$ we have $r_h(P) = \frac{q_h}{m}$ for some $q_h \in \mathbb{N}$. Then for each cluster $P(i)$ in a fair clustering $\mathcal{C}$ of $P$ with exact preservation of ratios, there exists a positive integer $i' \in \mathbb{N}_{\ge 1}$ such that $P(i)$ contains exactly $i' \cdot q_h$ points with color $h$ for each color $h \in Col$ and $i' \cdot m$ total points. Thus every cluster must have at least $q_h$ points of color $h$ for each color $h \in Col$.
\end{observation}

We use Observation~\ref{obs:fairletmultiples} and the fixed set of centers $S$ to obtain the following adjusted LP for the fractional fair $k$-center problem. 
\begin{align}
& \quad & \sum_{i \in S} x_{ij} &= 1, &\quad \forall j \in P\\
\label{eq:fair} & \quad &  \sum_{j \in col_h(P)} x_{ij} &= r_h(P)\sum_{j \in P} x_{ij} &\quad \forall i \in S\\
\label{eq:minreq} & \quad & \sum_{\substack{j \in col_h(P)\\ (i,j) \in E_{\tau}^2}} x_{ij} &\geq q_h &\quad \forall i \in S, \forall h \in Col\\
\label{eq:max_assign_dist}& \quad & x_{ij}  &= 0  &\quad \forall i \in S, j \in P \text{ with } (i,j) \notin E_{\tau}^3\\
& \quad & 0 \leq x_{ij}  & \leq 1  &\quad \forall i\in S, j \in P
\end{align}
Here inequality~\eqref{eq:minreq} ensures that each cluster contains at least $q_h$ points of color $h$. Let $S_{opt}$ be the set of centers in the optimal solution and let $\phi_{opt}: P \rightarrow S_{opt}$ be the optimal fair assignment. For the correct guess $\tau$, every center $i \in S$ has a distinct center in $S_{opt}$ which is at most distance one away from $i$ in $G_{\tau}$. Therefore, there exists $q_h$ points of each color $h$ within distance two of $i$. This ensures that inequality~\eqref{eq:minreq} is satisfiable for the right guess $\tau$. And since, every center in $S_{opt}$ is within distance two of some $i \in S$, there exists a fair assignment of points in $P$ to centers in $S$ within distance three. Thus the above LP is feasible for the right $\tau$.

Now for the final phase, the algorithm rounds a fractional solution for the above assignment LP to an integral solution of cost at most $5\tau$ in a procedure motivated by the LP rounding approach used by Cygan et al. in~\cite{CHK12} for the capacitated $k$-center problem. Let $\beta(i)$ denote the children of node $i \in S$ in the tree $T$. Define quantities $\Gamma(i)$ and $\delta(i)$, $\forall i \in S$ as follows:
\begin{align*}
\Gamma(i) &= \left\lfloor \frac{\sum_{j\in col_1(P)} x_{ij} + \sum_{i' \in \beta(i)} \delta(i')}{q_1} \right\rfloor q_1\\ 
\delta(i) &= \sum_{j\in col_1(P)} x_{ij} + \sum_{i' \in \beta(i)} \delta(i') - \Gamma(i)
\end{align*}
For a leaf node $i$ in the tree $T$ we have $\beta(i) = \emptyset$, then $\Gamma(i)$ denotes the amount of color $1$ points assigned to $i$ rounded down to the nearest multiple of $q_1$, while $\delta(i)$ denotes the remaining amount. The idea is to reassign the remainder to the parent of $i$. Then for a non leaf $i'$ $\Gamma(i')$ denotes the amount of color $1$ points assigned to $i'$ plus the remainder that all children of $i'$ want to reassign to $i'$ rounded down to the nearest multiple of $q_1$, while $\delta(i')$ again denotes the remainder.
Since by definition of $q_1$ the total number of points in $col_1(P)$ must be an integer multiple of $q_1$, $\Gamma(r)$ also denotes the the amount of color $1$ points assigned to $r$ plus the remainder that all children of $r$ want to reassign to $r$ and $\delta(r) = 0$.

Also note that $\Gamma(i)$ is always a positive integer multiple of $q_1$ for any $i$, and $\delta(i)$ is always non-negative and less than $q_1$. 

One can think of the $x_{ij}$ variables as encoding flow from a vertex $j$ to a node $i \in S$.
We call it a color $h$ flow if $j$ has color $h$. We will re-route these flows (maintaining the
ratio constraints) such that $\forall i \in S$, $j \in col_1(P)$ $x_{ij}$ is equal to $\Gamma(i)$ which is an integral multiple of $q_1$. 
\begin{lemma}
\label{lemma:integral_assignment_1}
There exists an integral assignment of all vertices with color $1$ to centers in $S$ in $G_{\tau}^5$ that assigns $\Gamma(i)$ vertices with color $1$ to each center $i \in S$.
\end{lemma}

\begin{proof}
Construct the following flow network: Take sets $col_1(P)$ and $S$ to form a bipartite graph with an edge of capacity one between a vertex $j \in col_1(P)$ and a center $i \in S$ if and only if $(i, j) \in E_{\tau}^5$. Connect a source $s$ with unit capacity edges to all vertices in $col_1(P)$ and each center $i \in S$ with capacity $\Gamma(i)$ to a sink $t$. We now show a feasible fractional flow of value $|col_1(P)|$ in this network. For each leaf node $i$ in $T$ which is not the root, assign $\Gamma(i)$ amount of color $1$ flow from the total incoming color $1$ flow $\sum_{j \in col_1(P)} x_{ij}$ from vertices that are at most distance three away from $i$ in $G_\tau$ and propagate the remaining $\delta(i)$ amount of color $1$ flow, coming from distance two vertices, upwards to be assigned to the parent of node $i$. This is always possible because by definition $\delta(i) < q_1$ and constraint~\eqref{eq:minreq} ensures that every center has at least $q_1$ amount of color $1$ flow coming from distance two vertices. For every non-leaf node $i$, assign $\Gamma(i)$ amount of incoming color $1$ flow from distance five vertices (including the color $1$ flows propagated upwards by its children) and propagate $\delta(i)$ amount of color $1$ flow from distance two vertices (possible due to constraint~\eqref{eq:minreq}). Thus every center has $\Gamma(i)$ amount of color $1$ flow passing through it and it is easy to verify that the value of the total flow in the network is $|col_1(P)|$. Since the network only has integral capacities, there exists an integral max-flow of value $|col_1(P)|$. 
\end{proof}

\begin{lemma}\label{lemma:other_color_reassignment}
For any reassignment of a color $1$ flow, there exists a reassignment of color $h$-flow between the same centers for all $h \in Col\setminus \{1\}$, such that the resulting fractional assignment of the vertices satisfies the fairness constraints at each center.
\end{lemma}
\begin{proof}
Say $f_1$ amount of color $1$ flow is reassigned from center $i_1$ to another center $i_2$. 
Reassign $f_h = r_h \cdot f_1/ r_1$ amount of color $h$ flow from $i_1$ to $i_2$ for each color $h \in Col \setminus\{1\}$. This is possible as constraint~\eqref{eq:fair} implies that the amount of color $h$ points assigned to $i_1$ must be equal to $\frac{r_h}{r_1}$ times the amount of color $1$ points assigned to $i_1$ and $f_1$ must be less than the amount of color $1$ points assigned to $i_1$.
It is easy to verify that the ratios at $i_1$ and $i_2$ remain unchanged as by construction the ratio of the reassigned flows is equal to the original ratio.
\end{proof}
From Lemmas~\ref{lemma:integral_assignment_1} and~\ref{lemma:other_color_reassignment} we can say that there is a fair fractional assignment within distance $5\tau$ such that all the color $1$ assignments are integral and every center $i$ has $\Gamma(i)$ color $1$ vertices assigned to it. 
Since this assignment is fair the total incoming color $h$ flow at each center must be $\Gamma(i) \frac{q_h}{q_1}$ which are integers for every center $i \in S$ and every color $h \in Col$.

\begin{lemma}
\label{lemma:integral_assignment}
There exists an integral fair assignment in $G_{\tau}^5$.
\end{lemma}

\begin{proof}
Construct a flow network for color $h$ vertices similar to the one in lemma~\ref{lemma:integral_assignment_1}: Take sets $col_h(P)$ and $S$ to form a bipartite graph with an edge of capacity one between a vertex $j \in col_h(P)$ and a center $i \in S$ if and only if $(i, j) \in E_{\tau}^5$. Connect a source $s$ with unit capacity edges to all vertices in $col_h(P)$  and each center $i \in S$ with capacity $\Gamma(i)\frac{r_h}{r_1}$ to a sink $t$. 
The above fractional assignment in $G_{\tau}^5$ gives a flow for the above network. Since the network only consists of integral demands and capacities, there is an integral max-flow which gives the assignment for the color $h$ vertices. 
\end{proof}

\begin{theorem}
\label{thm:5_approximation:kcenter}
There exists a $5$-approximation for the fair $k$-center problem with exact preservation of ratios.
\end{theorem}

\begin{proof}
Follows from Lemmas~\ref{lemma:integral_assignment_1},~\ref{lemma:other_color_reassignment} and~\ref{lemma:integral_assignment}
\end{proof}

\subsection{\texorpdfstring{$7$}{7}-approximation for \texorpdfstring{$k$}{k}-suppliers}
We adapt the algorithm in Section~\ref{sec:five_approx} to work for the $k$-suppliers model to give a 7-approximation for the variant with exact preservation of ratios. In the $k$-suppliers model, we are not allowed to open centers anywhere in $P$. Instead, we are provided a set $L$ of potential locations to open centers. The procedure closely resembles the $k$-center algorithm: construct a bipartite threshold graph $G_{\tau} = (P \cup L, E_{\tau})$ where $E_{\tau} = \{(i,j) \mid i \in L, j \in P, d(i,j) \leq \tau \}$. Choose a \emph{root} vertex $r \in P$ into $S$ and mark all vertices in $P$ that are within distance two. Until all vertices in $P$ are marked, choose an unmarked vertex $u \in P$ that is distance two away from a marked vertex and mark all vertices in the distance two neighborhood of $u$. Note that, since $G_{\tau}$ is bipartite, no two vertices in $P$ are adjacent. The vertex $u$ is exactly at distance four from a vertex $u' \in S$ chosen earlier. This process of selecting vertices in $S$ defines a tree $T$ over them with the property that adjacent vertices in $T$ are exactly at distance four of each other in $G_{\tau}$. Since we apply the procedure separately for each of the connected components of the threshold graph, we may safely assume that $G_{\tau}$ is connected.

Let us now temporarily open one center at each vertex in $S$ and make the following observations for the $k$-suppliers case:
\begin{enumerate}
\item Observation~\ref{obs:fairletmultiples} still holds.
\item The corresponding LP is the same as the $k$-center LP, except it has $E_{\tau}^4$ in place of $E_{\tau}^3$ in constraint~\eqref{eq:max_assign_dist}. This ensures the feasibility of the LP since every location in $L$ is at most distance three away from some vertex in $S$.
(Note that in case $G_\tau$ is not connected, it can happen that some locations in $L$ are not connected to any point and therefore more than distance three away from some vertex in $S$, but since they are not connected to any point we can safely ignore them, as they cannot be part of the optimal solution.)
\item Lemma~\ref{lemma:integral_assignment_1} with $G_{\tau}^6$ instead of $G_{\tau}^5$ holds. The extra distance of one is introduced because the distance between a child vertex and its parent vertex in $T$ is four instead of three.
\item Lemma~\ref{lemma:other_color_reassignment} holds as it is and Lemma~\ref{lemma:integral_assignment} holds when $G_{\tau}^5$ is replaced with $G_{\tau}^6$.
\end{enumerate}

Thus we have a distance six fair assignment to centers in $S$. However, this is not a valid solution for $k$-suppliers as $S \subseteq P$ and we are allowed to open centers only in $L$. So, we move each of these temporary centers to a neighboring location in $L$ to obtain a distance seven assignment.

\section{NP-hardness of the fair assignment problem for \texorpdfstring{$k$}{k}-center}\label{sec:nphardness:assignmentproblem}

In this section, we reduce the Exact Cover by 3-sets to the fair assignment problem for $k$-center. 
The input to the Exact Cover by 3-sets problem is a ground set $\mathcal{U}$ of elements and a family $\mathcal{F}$ of subsets such that each set has exactly three elements from $\mathcal{U}$. The objective is to find a set cover such that each element is included in exactly one set. For example, let $\mathcal{U} = \{a,b,c,d,e,f\}, \mathcal{F} = \{A = \{a,b,c\}, B = \{b,c,d\}, C = \{d,e,f\}\}$ be an instance. The set $\{A,C\}$ is an exact cover.
We call the problem of computing a cost-minimal fair assignment of points to given centers the \emph{fair assignment problem}. It exists once for every objective listed above. Even for $k$-center, the fair assignment problem is NP-hard. This can be shown by a reduction from Exact Cover by 3-sets,
 a variant of set cover. The input is a ground set $\mathcal{U}$ of elements and a family $\mathcal{F}$ of subsets such that each set has exactly three elements from $\mathcal{U}$. The objective is to find a set cover such that each element is included in exactly one set. For example, let $\mathcal{U} = \{a,b,c,d,e,f\}, \mathcal{F} = \{A = \{a,b,c\}, B = \{b,c,d\}, C = \{d,e,f\}\}$ be an instance. The set $\{A,C\}$ is an exact cover.

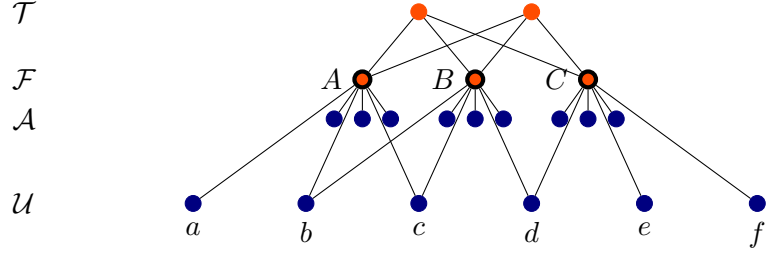
\begin{figure}
\centering
  \begin{tikzpicture}[scale=0.75,
    ]   
			% red points for F
			\node at (-5,1.2) {$\mathcal{F}$};
      \draw (1,1.2) node[rp,label=left:{$A$},draw=black,ultra thick] (f1) {};
      \draw (3,1.2) node[rp,label=left:{$B$},draw=black,ultra thick] (f2) {};
      \draw (5,1.2) node[rp,label=left:{$C$},draw=black,ultra thick] (f3) {};
			% blue points for U
			\node at (-5,-1) {$\mathcal{U}$};
      \draw (-2,-1) node[bp,label=below:{$a$}] (u1) {};
      \draw (0,-1) node[bp,label=below:{$b$}] (u2) {};
      \draw (2,-1) node[bp,label=below:{$c$}] (u3) {};
      \draw (4,-1) node[bp,label=below:{$d$}] (u4) {};
      \draw (6,-1) node[bp,label=below:{$e$}] (u5) {};
      \draw (8,-1) node[bp,label=below:{$f$}] (u6) {};
			% connect each set to its elements
			\draw (f1) -- (u1); \draw(f1) -- (u2); \draw(f1) -- (u3);
			\draw (f2) -- (u2); \draw(f2) -- (u3); \draw(f2) -- (u4);
			\draw (f3) -- (u4); \draw(f3) -- (u5); \draw(f3) -- (u6);		
			% auxiliary blue points for every set in F
			\node at (-5,0.5) {$\mathcal{A}$};
      \draw (0.5,0.5) node[bp] (f1a1) {};
      \draw (1,0.5) node[bp] (f1a2) {};
      \draw (1.5,0.5) node[bp] (f1a3) {};
			\draw (2.5,0.5) node[bp] (f2a1) {};
      \draw (3,0.5) node[bp] (f2a2) {};
      \draw (3.5,0.5) node[bp] (f2a3) {};
			\draw (4.5,0.5) node[bp] (f3a1) {};
      \draw (5,0.5) node[bp] (f3a2) {};
      \draw (5.5,0.5) node[bp] (f3a3) {};
			% connect auxiliary points
			\draw (f1) -- (f1a1);\draw (f1) -- (f1a2);\draw (f1) -- (f1a3);
			\draw (f2) -- (f2a1);\draw (f2) -- (f2a2);\draw (f2) -- (f2a3);
			\draw (f3) -- (f3a1);\draw (f3) -- (f3a2);\draw (f3) -- (f3a3);
		  % red points T
			\node at (-5,2.4) {$\mathcal{T}$};
			\draw (2,2.4) node[rp] (t1) {};
      \draw (4,2.4) node[rp] (t2) {};
			% connect T to everybody in F
			\draw (t1) -- (f1); \draw (t1) -- (f2); \draw (t1) -- (f3);
			\draw (t2) -- (f1); \draw (t2) -- (f2); \draw (t2) -- (f3);
\end{tikzpicture}
\caption{Example for the reduction from Exact Cover with 3-sets to the fair assignment problem for $k$-center, with $\mathcal{U}=\{a,b,c,d,e,f\}$ and $\mathcal{F} = \{A = \{a,b,c\}, B = \{b,c,d\}, C = \{d,e,f\}\}.$\label{fig:reduction}}
\end{figure}
For an instance $\mathcal{U}, \mathcal{F}$ of the exact cover problem, we construct an unweighted graph, which then translates to an input for the fair assignment problem for $k$-center by assigning distance $1$ to each edge and using the resulting graph metric.
The vertices consist of $\mathcal{U}$, $\mathcal{F}$ and two sets defined below, $\mathcal{A}$ and $\mathcal{F}$.
We start by adding an edge between all $e \in \mathcal{U}$ and any $A \in \mathcal{F}$ iff $e \in A$. 
We assign color red to the vertices from $\mathcal{F}$ and blue to those from $\mathcal{U}$. 
Then we construct a set $\mathcal{A}$ which contains three auxiliary blue vertices for each vertex in $\mathcal{F}$.
These are exclusively connected to their corresponding vertex in $\mathcal{F}$. 
Then we construct a set $\mathcal{T}$ of $|\mathcal{U}|/3$ red vertices.\footnote{Note that if $|\mathcal{U}|$ is not a multiple of three, it cannot have an exact cover, so we can assume that $|\mathcal{U}|$ is a multiple of three.} and connect each vertex in $\mathcal{T}$ to every vertex in $\mathcal{F}$. Finally, we open a center at each vertex in $\mathcal{F}$. The construction is shown in Figure~\ref{fig:reduction}.
Observe that the distance between an element $e \in \mathcal{U}$ and an open center at $A\in \mathcal{F}$ in this construction is $1$ iff $e \in A$, and otherwise, it is $3$: If $e \notin A$, then there is no edge between $e$ and $A$, and since there are no direct connections between the centers, the minimum distance between $e$ and another open center is $3$. 

\begin{lemma}\label{lemma:reduction:exactcovertofairassignment}
If there exists an exact cover, there exists a fair assignment of cost $1$ where the red:blue ratio is 1:3 for each cluster.
\end{lemma}

\begin{proof}
Assign each red vertex $A \in \mathcal{F}$ and the three auxiliary blue vertices connected to it to the center at $A$. If $A$ is in the exact cover, assign the three blue vertices representing its elements and one red vertex from $\mathcal{T}$ to the center at $A$. It is straightforward to verify that this assignment is fair and assigns every vertex to some center to which it is connected via a direct edge.
\end{proof}

\begin{lemma}\label{lemma:reduction:fairassignmenttoexactcover}
If there exists a fair assignment where red:blue = 1:3 for all clusters of cost less than $3$, there exists an exact cover.
\end{lemma}

\begin{proof}
For $A \in \mathcal{F}$, the red vertex at $A$ and the three auxiliary blue vertices attached to it must be assigned to the center at $A$ as this is the only center within distance less than $3$. Also, no center can have more than two red vertices assigned to it because there are only six blue vertices in distance less than $3$ of any center. Therefore, each red vertex in $\mathcal{T}$ must be assigned to a distinct center and each such center $A$ will have exactly three blue vertices from $\mathcal{U}$ assigned to it which correspond to the elements in the set that $A$ represents. Thus, the sets corresponding to the centers that have two red vertices assigned to them form an exact cover for $\mathcal{U}$.
\end{proof}

\section{Integrality gap of the canonical clustering LP}\label{sec:appendix-integrality-gap}
We show that any integral fair solution needs large clusters to implement awkward ratios of the input points. 
This allows us to derive a non-constant integrality gap for the canonical clustering LP.
\begin{lemma}\label{lem:num-cluster-bound}
  Let $P$ be a point set with $r$ red and $r-1$ blue points and let $k \geq 1$.
  If the ratio of red points $r_{red}(C_i)$ is at most $\frac{r-k+1}{2r-2k+1}$ for each cluster $C_i$, then any fair solution can have at most $k$ clusters.
\end{lemma}
\begin{proof}
  Consider a solution with $k' > k$ clusters.
  Since we have more red points there must be at least one cluster $C_i$ that contains more red points than blue points. 
  The ratio of red points $r_{red}(C_i)$ of this cluster is minimized if the solution contains $k'-1$ clusters with one blue and one red point, and one cluster with the remaining $r-k'$ blue and $r-k' +1$ red points.
  However, 
  \begin{align*}
    \frac{r-k'+1}{2r -2k' + 1} > \frac{r-k+1}{2r -2k + 1} 
  \end{align*}
  Since the highest ratio of red points in any other solution can only be higher, the claim follows.
\end{proof}
We remark that Lemma~\ref{lem:num-cluster-bound} is not true for essentially fair solutions.

The canonical fair clustering ILP consists of~\eqref{eq:general:lppointassigned}--\eqref{eq:obj:kcenter} and~\eqref{eq:general:fair}.
In the $k$-median/facility location case and in the $k$-means case, let write $\OPT_F$ for the optium value of its LP relaxation and and let us call the value of an optimum integral solution $\OPT_I$. 
We then define the integrality gap of the ILP as $\OPT_I / \OPT_F$.
In the $k$-center case, the ILP does not have an objective function, but we can define its integrality gap in the following sense: If $\tau_I$, $\tau_F$ is the smallest $\tau$ such that the LP-relaxation has a feasible \emph{integral} or \emph{fractional} solution, respectively, then we define the integrality gap as $\tau_I / \tau_F$.
\lemintgap*
\begin{proof}
\begin{figure}[ht]
\centering
  \includegraphics[width=13.5cm]{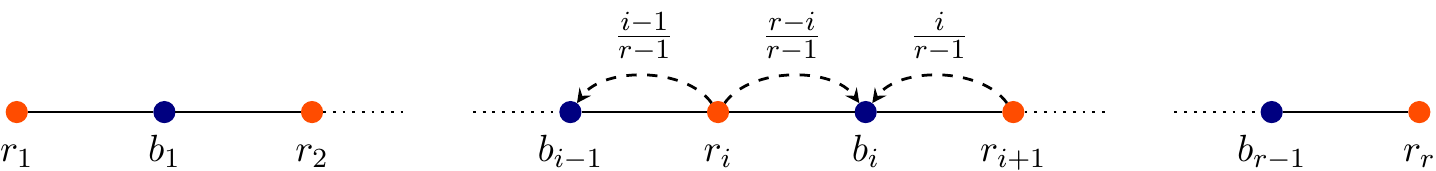}
 \caption{Integrality gap example.} \label{integralityGap}
\end{figure}
Consider the input points $P$ lying on a line, as shown in Figure ~\ref{integralityGap}. Specifically, we have $r$ red points $\{ r_1, r_2, \ldots, r_{r} \}$ that alternate with $r-1$ blue points $\{ b_1, b_2, \ldots, b_{r-1} \}$. The distance between consecutive points is $1$. 

We require that the ratio of the red points of each cluster is between $0$ and $(r-1)/(2r-3)$ and set $k = r-1$.
The input ratio $r/(2r-1)$ of the red points lies in the interior of this interval as
\begin{align*}
      \frac{r}{2r-1} < \frac{r-1}{2r-3} \iff 2r^2 - 3r < 2r^2 - 3r + 1,
\end{align*}
and thus our input is well-defined and the fairness relaxation is non-trivial.
We then ask for a clustering of $P$ with at most $k$ centers that respects the fairness constraints. 

Consider the following feasible solution for the LP-relaxation.
The solution opens a center at each of the $r-1=k$ blue points and assigns the blue point to itself and the red points on each side in the following way: for each $1 \le i \le r-1$, assign $r_i$ to $b_i$ by a fraction of $\frac{r - i} {r-1}$ and for each $2 \leq i \leq r$ assign $r_i$ to $b_{i-1}$ a fraction of $\frac{i-1}{r-1}$. 
Each red point is fully assigned in this way. 
We also get that in a cluster around some fixed $b_i$, the total assignment coming from red points is $\frac{r}{r-1}$  and the assignment coming from blue points is $1$; thus, each cluster has a ratio of red points of 
\begin{align*}
\frac{\frac{r}{r-1}}{1 + \frac{r}{r-1}} = \frac{\frac{r}{r-1}}{\frac{2r - 1}{r-1}} = \frac{r}{2r-1}.
\end{align*}
We therefore respect the balance requirements. 

However, as $(r-1)/(2r -3) = (r- k' + 1)/(2r -2k' + 1)$ for $k' = 2$, by Lemma~\ref{lem:num-cluster-bound} any integral solution satisfying the ratio requirement can at most open two centers.

\begin{itemize}
\item In the $k$-center case, the fractional solution has a radius of $1$ and the integral solution has a radius of at least $\lfloor{(r-1)/2}\rfloor = \Omega(n)$.
The $k$-center problem is a special case of the $k$-supplier problem; thus, the integrality gap for the $k$-supplier problem can only be larger.
\item In the $k$-median case, the fractional solution has a cost of $O(n)$: The blue points incur no cost and each red point $r_i$ contributes $(r-i) / (r-1)\cdot 1 + (i-1) /(r-1) \cdot 1 = 1$ to the objective function. 
Since the optimum integral solution can have at most two centers, it has to contain one cluster spanning at least $\lfloor r/2\rfloor$ consecutive points. This incurs a cost of at least $2\cdot\sum_{j=1}^{\lfloor r/4\rfloor-1} j = \Omega(n^2)$. 
\item In the facility location case, we observe that we can open at most two facilities in a fair integral solution.
Hence, the analysis for the $k$-median case carries over (even if we set all opening costs to zero).
\item In the $k$-means case, each red point $r_i$ incurs a cost of $(r-i) / (r-1)\cdot 1^2 + (i-1) /(r-1) \cdot 1^2 = 1$ in the fractional solution; the blue points again incur no cost as they are chosen as centers.
However, the integral solution now has a cost of at least $2\cdot\sum_{j=1}^{\lfloor r/4\rfloor-1} j^2 = \Omega(n^3)$.
\end{itemize}
\end{proof}
This integrality gap yields a lower bound on the quality guarantee of any LP-rounding approach for this ILP.
Thus, Lemma~\ref{lem:lp-integrality-gap} implies that no fair constant factor approximation for can be achieved by rounding the canoncial fair clustering ILP.
The counterexample in~\ref{lem:lp-integrality-gap} breaks down in the essential fairness model.

\bibliographystyle{plainurl}
\bibliography{references}

\appendix

\newpage
\section{Notation}

\begin{table}[htb]
 \centering
 \begin{tabular}{ll}
  \textbf{Object}  & \textbf{Notation} \\
	\toprule
  the set of points  & $P$ \\
  the set of possible locations & $L$ \\
  an individual point & $j$ or $p_j$\\
	an individual center & $i$ or $c_i$\\
	the number of points & $n$ \\
  the set of chosen locations/ the set of centers & $S$ \\
  an assignment of points to centers & $\phi: P \rightarrow S$ \\
	a clustering & $\mathcal{C} = (S,\phi)$ \\
	a cluster in a clustering & $P(i) = \{p \in P \mid \phi(p) = c_i\}$ \\
  the set of colors & $Col$ \\
  an individual color & $h$ or $col_h$ \\
  color assignment & $col: P \rightarrow Col$ \\
	points with color $col_h$ in $P' \subseteq P$ & $col_h(P') = \{j \in P' \mid col(j) = col_h\}$ \\
	the amount of points with color $col_h$ in $P' \subseteq P$ & $\mass_h(P') = |col_h(P')|$\\
	the amount of color $col_h$ assigned to $i$ in an (I)LP solution $(x,y)$ & $\mass_h(x,i) = \sum_{j\in col_h(P)} x_{ij}$\\
	the amount assigned to $i$ in an (I)LP solution $(x,y)$ & $\mass(x,i) = \sum_{j\in P} x_{ij}$\\
	the ratio of points with color $col_h$ in a set $P' \subseteq P$ & $r_h (P')  = \frac{|col_h(P')|}{|P'|}$\\
	the number of allowed clusters & $k$ \\
  the distance function & $d: (P\cup L) \times (P\cup L) \rightarrow \mathbb{R}_{\ge 0}$ \\
	opening costs & $f_i$ \\
  relaxed boarders for the ratio & $\ell_h= \frac{p_1^h}{q_1^h},u_h = \frac{p_2^h}{q_2^h} \in \mathbb{Q}_{\ge 0}$ \\
  a fairlet decomposition & $\mathcal{F}$ \\
  a fairlet & $F_i$ \\
	threshold & $\tau$ \\
	threshold graph with threshold $\tau$ on $P$ & $G_{\tau} = (P \cup L,E_{\tau})$\\
	edges in the threshold graph & $E_{\tau} = \{(i,j) \mid i \in L, j \in P, \dist(i,j) \le \tau\}$ \\ 
	$k$-center without fairness constraint & standard $k$-center\\
	(similar for other objectives) \\
	
	\bottomrule
 \end{tabular}
 \caption{Objects and their notation\label{table:notation}}
\end{table}

\section{Further related work}\label{sec:furtherrelatedwork}
Using $k$ centers to cluster points while minimizing a certain objective function has a long history in terms of results and applications. 
For the $k$-center problem in general metric spaces, the $2$-approximations developed by Gonzalez~\cite{G85} and Hochbaum and Shmoys~\cite{HS86} were shown to be tight by Hsu and Nemhauser~\cite{HN79}. 
The $k$-supplier problem can be $3$-approximated~\cite{HS86}, which is also tight.
Facility location can be $1.488$-approximated~\cite{L13}, which is very close to the known APX-hardness of $1.463$ for the problem~\cite{GK99}. For $k$-median, a recent breakthrough has led to a $2.765$-approximation~\cite{LS16,BPRST17}, while the best hardness result lies  below two~\cite{JMS02}.
The gap between best upper and lower bound is even larger for $k$-means, where a $6.357$-approximation is the best known~\cite{ANSW17}, and the newest hardness result is marginally above 1~\cite{ACKS15,LSW17}.

The wide applicability of $k$-center has given rise to numerous variants that strive to incorporate useful constraints into the solution, such as capacity constraints~\cite{BKP93,CHK12,KS00}, lower bounds on the size of each cluster~\cite{APFTKKZ10,AS16} or allowing for outliers~\cite{CKMN01,CK14}. 
The facility location problem has also been studied under capacity constraints~\cite{ALBGGGJ13,ASS17,BGG12}, with uniform lower bounds~\cite{AS12,S10}, and with outliers~\cite{CKMN01}.
Much less is known for $k$-median and $k$-means. True constant-factor approximations so far exist only for the outlier constraint~\cite{C08,KLS18}, and are not known for the variants with capacities or lower bounds. A major problem for obtaining constant factor approximations is that the natural LP for the problem has an unbounded integrality gap, which as we will see is also true for the LP with fairness constraints. 
Bicriteria approximations are known that either violate the capacity constraints~\cite{Li14,Li16,Li17} or the cardinality constraint~\cite{ABGL15}.

\section{Facts about the \texorpdfstring{$k$}{k}-means cost function}

We use some well-known facts about the $k$-means function when extending our results for $k$-median to $k$-means. The first one is that squared distances satisfy a relaxed triangle inequality:

\begin{lemma}\label{lemma:kmeans2rti}
It holds for all $x, y, z \in \R^d$ that
\[
||x-z||^2 \le 2 ||x-z||^2 + 2 ||z-y||^2.
\]
\end{lemma}

The next lemma is also a folklore statement which can be extremely useful. It implies that the best $1$-means is always the centroid of a point set, and has further consequences, like Lemma~\ref{lemma:twofactorkmeans} which we state below, a fact which is also commonly used in approximation algorithms for the $k$-means problem. 

\begin{lemma}\label{lemma:magicformula}
For any $P \subset \R^d$, and $z \in \R^d$, 
\[
\sum_{x \in P} ||x-z||^2 = \sum_{x\in P} ||x - \mu(P)||^2 + |P| \cdot ||\mu(P)-z||^2,
\]
where $\mu(P) = \frac{1}{|P|} \sum_{x \in P} x$ is the centroid of $P$.
\end{lemma}

One corollary of Lemma~\ref{lemma:magicformula} is that the optimum cost of the best discrete solution is not much more expensive than the best choice of centers from $\R^d$.

\begin{lemma}\label{lemma:twofactorkmeans}
Let $P\subset \R^d$ be a set of point in the Euclidean space, and let $C^\ast \subset \R^d$ be a set of $k$ points that minimizes the $k$-means objective, i.e., it minimizes
\[
\sum_{x \in P} \min_{c \in C} ||x-c||^2
\]
over all choices of $C \subset \R^d$ with $|C|=k$.
Furthermore, let $\hat{C}$ be the set of centers that minimizes the $k$-means objective over all choices of $C \subset P$ with $|C|=k$, i.e., the best choice of centers from $P$ itself.
Then it holds that
\[
\sum_{x \in P} \min_{c \in \hat{C}} ||x-c||^2 \le 
\sum_{x \in P} \min_{c \in C^\ast} ||x-c||^2.
\]
Thus, restricting the set of centers to the input point set increases the cost of an optimal solution by a factor of at most $2$.
\end{lemma}
\section{Computing fair clusterings via fairlets}\label{sec:fairletapproach}

\subsection{Fairlets and fairlet decompositions} 
Chierichetti et al.~\cite{CKLV17} make an important observation: 
Any cluster of a fair clustering is composed of \emph{atomic} subclusters. 
For instance, if we are to achieve a ratio of 2/3 of red to blue points, then a cluster with two red and three blue cannot be divided into fair subclusters.
Moreover, any fair cluster with an overall balance ratio of 2/3 consists of an integer number of atomic subclusters -- each having two red and three blue points.
Figure~\ref{fig:fairlet-decomp-example} visualizes this concept. 
We follow~\cite{CKLV17} and call these atomic microclusters \emph{fairlets}.
The concept naturally generalizes to clusterings under relaxed fairness.
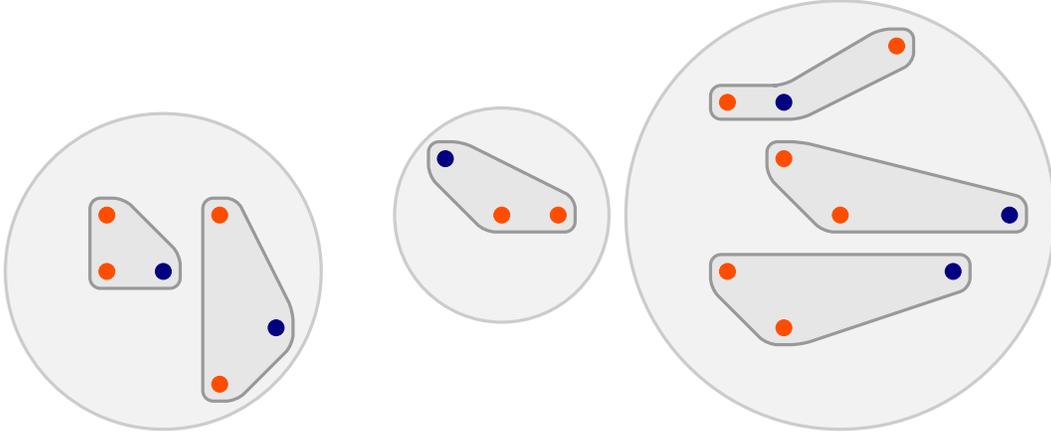
\begin{figure}
\centering
  \begin{tikzpicture}[scale=0.75,
    ]   
      % Predefine coordinates for the points in the first cluster
      \coordinate (r11) at (0,2);
      \coordinate (r12) at (0,3);
      \coordinate (r13) at (2,0);
      \coordinate (r14) at (2,3);
      \coordinate (b11) at (1,2);
      \coordinate (b12) at (3,1);

      % cluster 1 with radius > sqrt(5)
      \draw[cluster] (b11) circle[radius=2.8cm];    
      
      % fairlets, first cluster
      \draw[fairlet]      
            ($(r11)-(0.3, 0.3)$) -- ($(r12)+(-0.3,0.3)$) 
        -- ($(r12)+(0.3,0.3)$)  -- ($(b11)+(0.3,0.3)$) 
        -- ($(b11)+(0.3,-0.3)$) -- cycle
      ;
      \draw[fairlet] 
            ($(r13)-(0.3, 0.3)$) -- ($(r14)+(-0.3,0.3)$)
        --  ($(r14)+(0.3, 0.3)$) -- ($(b12)+(0.3,0.3)$)
        --  ($(b12)+(0.3, -0.3)$) -- ($(r13)+(0.3, -0.3)$)
        -- cycle
      ;

      % Predefine coordinates for the points in the second cluster
      \coordinate (r21) at (7,3);
      \coordinate (r22) at (8,3);
      \coordinate (b21) at (6,4);

      % cluster 2 with radius > sqrt(2)
      \draw[cluster] (r21) circle[radius=1.9cm];
      
      % Fairlet, Cluster 2
      \draw[fairlet]
            ($(r21)-(0.3, 0.3)$)  -- ($(b21)-(0.3,0.3)$) 
         -- ($(b21)+(-0.3, 0.3)$) -- ($(b21)+(0.3,0.3)$)
         -- ($(r22)+(0.3, 0.3)$) -- ($(r22)+(0.3, -0.3)$)
         -- cycle
      ;

      % Predefine coordinates for points of cluster 3
      \coordinate (r31) at (11,2);
      \coordinate (r32) at (12,1);
      \coordinate (r33) at (12,4);
      \coordinate (r34) at (13,3);
      \coordinate (r35) at (11,5);
      \coordinate (r36) at (14,6);
      \coordinate (b31) at (15,2);
      \coordinate (b32) at (16,3);
      \coordinate (b33) at (12,5);      
      
      % cluster 3 with radius sqrt(9)
      \draw[cluster] (r34) circle[radius=3.8cm];      
      
      % Fairlets, Cluster 3
      \draw[fairlet] 
           ($(r31) -(0.3, 0.3)$)  -- ($(r31) +(-0.3, 0.3)$)
        -- ($(b31) +(0.3,0.3)$)   -- ($(b31) +(0.3, -0.3)$)
        -- ($(r32) +(0.3,-0.3)$)   -- ($(r32) -(0.3, 0.3)$)
        -- cycle
      ;
      \draw[fairlet]
           ($(r33) +(-0.3, 0.3)$) -- ($(r33) +(0.3, 0.3)$)
        -- ($(b32) +(0.3, 0.3)$)  -- ($(b32) +(0.3, -0.3)$)
        -- ($(r34) -(0.3, 0.3)$)  -- ($(r33) -(0.3, 0.3)$)
        -- cycle
      ;
      \draw[fairlet]
           ($(r35) -(0.3,0.3)$)  -- ($(b33) +(0.3, -0.3)$)
        -- ($(r36) +(0.3,-0.3)$) -- ($(r36) +(0.3, 0.3)$)
        -- ($(r36) +(-0.3,0.3)$) -- ($(b33) +(0.0, 0.3)$)
        -- ($(b33) +(-0.3,0.3)$) -- ($(r35) +(-0.3, 0.3)$)
        -- cycle
      ;

      % draw actual points of cluster 1
      \draw (r11) node[rp] {};
      \draw (r12) node[rp] {};
      \draw (r13) node[rp] {};
      \draw (r14) node[rp] {};
      \draw (b11) node[bp] {};
      \draw (b12) node[bp] {};
      
      % draw actual points of cluster 2
      \draw (r21) node[rp] {};
      \draw (r22) node[rp] {};
      \draw (b21) node[bp] {};

      % draw actual points of cluster 3
      \draw (r31) node[rp] {};
      \draw (r32) node[rp] {};
      \draw (r33) node[rp] {};
      \draw (r34) node[rp] {};
      \draw (r35) node[rp] {};
      \draw (r36) node[rp] {};
      \draw (b31) node[bp] {};
      \draw (b32) node[bp] {};
      \draw (b33) node[bp] {};            
      
\end{tikzpicture}
\caption{\label{fig:fairlet-decomp-example}A decomposition of a fair clustering into fairlets with two red points~(\RedNode) and one blue point~(\BlueNode).}
\end{figure}
\begin{definition}
A set of points $F \subseteq P$ is called a \emph{fairlet}, if $F$ is balanced and cannot be partitioned into multiple, non-empty balanced sets.
Given $l = (l_0,\ldots,l_g)$ and $u = (u_0,\ldots,u_g)$ a set $F \subseteq P$ is called an $l,u$-fairlet, if $F$ is $l,u$-balanced and can not be partitioned into multiple, non-empty $l,u$-balanced sets.
\end{definition}
When we refer to fair clusters and fairlets we also mean $l,u$-fair and $l,u$-fairlets (unless stated otherwise) as most of the proof are analogous.
As observed above, every fair cluster can be partitioned into fairlets.
The converse is also true: The union of disjoint fairlets is a balanced clustering.

\begin{definition}
Let $P = \bigcup_{i\in I} F_i$ be a partition of $P$ such that for each $i \in I$ the set $F_i \subseteq P$ is a fairlet (an $l,u$-fairlet). Then we call $\mathcal{F} = \{F_i \mid i \in I\}$ a \emph{fairlet decomposition} (an $l,u$-fairlet decomposition).
\end{definition}
An intuitive way to look at a fairlet decomposition $\mathcal{F}$ is to view it as a fair clustering of $P$ with $k'=|I|$ clusters (where some fairlets might share a center).\footnote{Obtaining a fairlet decomposition from a $k'$-clustering requires additional work, however: If the clustering has strictly less than $k'$ centers, we need to subdivide some clusters into fairlets with a shared center.}   
In the following, we denote the center of fairlet $F_i \in \mathcal{F}$ in this clustering by $c_i$.
We then assign a cost to $\mathcal{F}$ in the following way
\begin{align*}
cost(\mathcal{F}) &=  
\begin{cases}
\max_{i \in I} \max_{p \in F_i} d(p,c_i)   & k-\text{center/supplier problem} \\
\sum_{i \in I} \sum_{p \in F_i} d(p,c_i)   & k-\text{median problem}/\text{facility location problem} \\
\sum_{i \in I} \sum_{p \in F_i} d(p,c_i)^2 & k-\text{means problem,} 
\end{cases}\\
&=
\begin{cases}
\max_{i \in I} \min_{c \in L} \max_{p \in F_i} d(p,c) & k-\text{center/supplier problem} \\
\sum_{i \in I} \min_{c \in L} \sum_{p \in F_i} d(p,c)   & k-\text{median problem}/\text{facility location problem} \\
\sum_{i \in I} \min_{c \in L} \sum_{p \in F_i} d(p,c)^2 & k-\text{means problem,} 
\end{cases}
\end{align*}
and call the problem to find a fairlet decomposition of minimal cost the \emph{fairlet decomposition problem}.
This assignment of costs is natural: The cost of $\mathcal{F}$ is exactly the cost of the corresponding $k'$-clustering.
It follows that the cost $\OPTfl$ of an optimum fairlet decomposition -- i.e., the cost of an optimum fair clustering of $P$ with $k' \geq k$ centers -- cannot be larger than the cost $\OPT$ of an optimum fair clustering with $k$ centers. 
In other words, we have $\OPTfl \leq \OPT$ for the fair $k$-center/supplier problem, for the fair $k$-median problem, for the fair $k$-means problem, and for the fair facility location problem. 

\subsection{The general fairlet approach for different problems}
To see why fairlets help us to compute fair clusterings let us first consider a simple example.
Suppose we are looking for an optimum fair $k$-center clustering of a point set $P$ and denote the cost of this clustering by $\OPT$.
Also suppose that we know an optimum fairlet decomposition $\mathcal{F}:=\{F_i \mid i \in I\}$ of $P$ with cost $\OPTfl$ that assigns a center $c_i$ to each $F_i\in \mathcal{F}$.
How can we use this knowledge to find a good clustering of $P$?
Since the union of disjoint fairlets yields a fair clustering, it seems a good idea to contract each fairlet to a single point and to then compute a colorblind clustering. 
This idea works indeed and was first proposed in~\cite{CKLV17}:
We interpret $c_i$ as a representative of $F$ and compute an $\gamma$-approximate colorblind $k$-center clustering $\mathcal{C}'$ of the representatives $P'=\{c_i \mid i \in I\}$. 
Let us say that $\mathcal{C}'$ assigns $c_i$ to a center $c(c_i) \in L$.
Then, we merge the fairlets that were assigned to the same center into a fair cluster: 
We obtain a fair clustering $\mathcal{C}$ by assigning the points $p \in F_i$ of each fairlet $F_i \in \mathcal{F}$ to $c(c_i)$.
As a union of disjoint fairlets each cluster of $\mathcal{C}$ must be fair. 
Also, by our previous observation that $\OPTfl \leq \OPT$, each fairlet must have a radius of at most $\OPT$.
Likewise, we can bound costs of the colorblind clustering $\mathcal{C}'$ by $\OPT$: 
An optimum colorblind clustering cannot be more expensive than an optimum fair clustering. 
Thus, the radius of each cluster in $\mathcal{C}$ can at most be $(1+\gamma)\cdot \OPT$.
We have found a $(1+\gamma)$-approximation!

In the remainder of this section, we explore what kind of approximation ratios we can expect when we cluster fairlets with a colorblind algorithm. 
To do so, we formalize the approach and generalize it to the $k$-supplier, the $k$-median and the $k$-means problem where we will need more advanced techniques for the analysis. 

\subsection{Using fairlets to obtain clusterings}
As before, we look for clusterings of a pointset $P$ and denote the cost of an optimum (exact or relaxed) fair $k$-clustering by $\OPT$ (we rely on the context to make clear whether this means $k$-center, $k$-median or $k$-means).
We let $\OPTfl$ denote the cost of an optimum decomposition of $P$ into $k'$ fairlets and we denote the cost of an optimum colorblind $k$-clustering of $P$ by $\OPTcb$.
Since any fair $k$-clustering is a valid $k'$-clustering we have $\OPTfl \leq \OPT$. 
Likewise, any fair $k$-clustering is a valid colorblind $k$-clustering and $\OPTcb \leq \OPT$ follows as well; just as in the above example.

In the sequel, we suppose that we can compute an $\alpha$-approximate fairlet decomposition $\mathcal{F} = \{F_i \mid i \in I\}$ of $P$, i.e. we suppose that $\cost(\mathcal{F}) \leq \alpha\OPTfl \leq \alpha \OPT$. 
We also suppose that we can compute a $\gamma$-approximate colorblind $C'$ approximation, i.e. we ask that $\cost(C') \leq \gamma \OPTcb \leq \gamma \OPT$. 
The current state-of-the-art does not provide non-trivial approximations in all cases, but we postpone this issue until the next section. 
For now, let us assume that we have the necessary blackbox algorithms and explore the resulting guarantees of the fairlet approach.

\paragraph{\kk-center.}
We repeat the argument from the introductory example: Given an $\alpha$-approximate fairlet decomposition of $P$ that assigns a center $c_i$ to each $F_i$ in $\mathcal{F}$, we compute a $\gamma$-approximate colorblind $k$-center solution $\mathcal{C}'$ of $P':=\{c_i \mid i \in I\}$.
If $c(c_i) \in L$ is the center assigned to $c_i$ in $\mathcal{C}'$, we obtain a fair clustering $\mathcal{C}$ of $P$ by assigning all $p \in F_i$ to $c(c_i)$.
We can now bound the distance of any point $p \in F_i$ to its center $c(c_i)$ in $\mathcal{C}$ as $d(p, c(c_i)) \leq d(p, c_i) + d(c_i, c(c_i))$ since $d$ is a metric. 
Yet, we have $d(p,c_i) \leq \alpha\OPTfl \leq \alpha \OPT$ and $d(c_i, c(c_i)) \leq \gamma \OPTcb \leq \gamma \OPT$. 
It follows that $\cost(C) \leq (\alpha+\gamma)\OPT$ and our algorithm is a $(\alpha+\gamma)$-approximation.

\paragraph{\kk-supplier.}
Given an $\alpha$-approximate fairlet decomposition of $P$ that assigns a center $c_i$ to each $F_i$ in $\mathcal{F}$, we cannot just compute a $\gamma$-approximate colorblind $k$-center solution $\mathcal{C}'$ of $P':=\{c_i \mid i \in I\}$ as we do not necessarily have $P' \subseteq P$. Instead of $c_i$ we choose an arbitrary point $p_i \in F_i$ as a representative for $F_i$ and compute a $\gamma$-approximate colorblind $k$-center solution $\mathcal{C}'$ of $P':=\{p_i \mid i \in I\}$. 
If $c(p_i) \in L$ is the center assigned to $p_i$ in $\mathcal{C}'$, we obtain a fair clustering $\mathcal{C}$ of $P$ by assigning all $p \in F_i$ to $c(p_i)$.
We can now bound the distance of any point $p \in F_i$ to its center $c(p_i)$ in $\mathcal{C}$ as $d(p, c(p_i)) \leq d(p, p_i) + d(p_i, c(p_i))$ since $d$ is a metric. 
Yet, we have $d(p,p_i) \leq 2\alpha\OPTfl \leq 2\alpha \OPT$ and $d(p_i, c(p_i)) \leq \gamma \OPTcb \leq \gamma \OPT$. 
It follows that $\cost(C) \leq (2\alpha+\gamma)\OPT$ and our algorithm is a $(2\alpha+\gamma)$-approximation. Some of the algorithms for the fairlet decomposition will already assume choose a center as one of the points. In that case it follows that $\cost(C) \leq (\alpha+\gamma)\OPT$ and our algorithm is a $(\alpha+\gamma)$-approximation.

\paragraph{\kk-median.}
The \kk-median case is more difficult: Computing representatives and repeating the above analysis yields an additive error of $\alpha\OPT$ \emph{for each point}, and thus at best a $2\cdot\alpha\cdot|P|$-approximation. 
A different technique is needed here.

We start by computing a $\gamma$-approximate colorblind clustering $\mathcal{C}'$ on $P$ and denote the center assigned to $p \in P$ in $\mathcal{C}'$ by $c(p)$. %\todo{Daniel: Now $C \not=L$ is ok?} 
For all $i \in I$ and purely for the analysis we choose the center $c_i$ of $F_i$ in $\mathcal{F}$ as a representative of $F_i$.
We then look for the point 
\begin{align*}
  p_i := \arg\min_{p \in F_i} \bigl[ d(c_i,p) + d(p, c(p))\bigr]
\end{align*}
that among all $p \in F_i$ minimizes the distance from $c_i$ to $c(p)$ via $p$.
This point $p_i \in P$ must have a center $c(p_i)$ in the colorblind clustering $\mathcal{C}'$ and we assign all points $p \in F_i$ to this center $c(p_i)$.
In this way, we obtain a fair clustering $\mathcal{C}$.
Let us first bound the cost of assigning a single point $p \in F_i$ to its center $c(p_i)$ in $\mathcal{C}$.
We have
\begin{align*}
  d(p, c(p_i)) \leq d(p, c_i) + d(c_i, p_i) + d(p_i, c(p_i))
\end{align*}
By our choice of $p_i$, we have $d(c_i, p_i) + d(p_i, c(p_i)) \leq d(c_i,p) + d(p, c(p))$ and it follows that 
\begin{align*}
  \cost(\mathcal{C}) &= \sum_{i \in I} \sum_{p \in F_i} d(p, c(p_i))\\
                     &\leq \sum_{i \in I} \sum_{p \in F_i} \bigl(2d(c_i, p) + d(p, c(p)) \bigr)\\
                     &= 2\cdot \sum_{i \in I} \sum_{p \in F_i} \bigl(d(p, c_i) + \sum_{p \in P} d(p, c(p))\\
                     &\leq 2\cdot \alpha \OPTfl + \gamma\OPTcb.
\end{align*}
Thus, our algorithm yields a $(2\alpha + \gamma)$-approximation.

\paragraph{Facility location.}
Observe that in the above algorithm, we only open the centers from the colorblind solution; the fairlet decomposition does not incur facility opening costs. 
By the same reasoning as before we get
\begin{align*}
  \cost(\mathcal{C}) &= \sum_{i \in I} \sum_{p \in F_i} d(p, c(p_i)) + \sum_{l \in L_{\text{cb}}} f_l\\
                     &\leq 2\cdot \sum_{i \in I} \sum_{p \in F_i} \bigl(d(p, c_i) + \sum_{p \in P} d(p, c(p)) + \sum_{l \in L_{\text{cb}}} f_l\\
                     &\leq 2\cdot \alpha \OPTfl + \gamma\OPTcb.
\end{align*}
where $L_{\text{cb}}$ denotes the set of facilities opened by the colorblind facility location solution.
We obtain a $(2\alpha + \gamma)$-approximation in this case as well.

\paragraph{\kk-means.}
For the $k$-means problem, we assume that $P \subset \R^n$ and that $d$ is the squared Euclidean distance measure.
For each fairlet $F_i$ let $\mu(F_i)$ be its centroid, i.e. $\mu(F_i) = \frac{1}{|F_i|}\sum_{p \in F_i} p$ and let us look at a colorblind $\gamma$-approximative clustering $\mathcal{C}'$ on the instance which contains $|F_i|$ copies of $\mu(F_i)$ for each $i \in I$. 
Without loss of generality, we assume that all of the copies of $\mu(F_i)$ will be assigned to the same center for all $i \in I$ as otherwise assigning all copies of $\mu(F_i)$ to the nearest opened center would only decrease the cost.
We then obtain $\mathcal{C}$ from $\mathcal{C}'$ by assigning all points in $F_i$ to the center $c_{F_i}$ to which the copies of $\mu(F_i)$ were assigned in $\mathcal{C}'$.
The following identity helps us to determine the cost $cost(\mathcal{C})$ of this assignment.

\begin{proposition}
\label{prop:centroid}
Given a point set $P \subseteq \mathbb{R}^d$ and a point $c \in \mathbb{R}^d$, then the $1$-means cost of clustering $P$ with $c$ can be decomposed into
\[\sum_{p \in P} \norm{p - c}^2 = \sum_{p \in P} \norm{p- \mu}^2 + |P| \cdot \norm{\mu - c}^2\]
where $\mu = \frac{1}{|P|} \sum_{p\in P} p$ is known as the centroid of $P$.
\end{proposition}

Proposition~\ref{prop:centroid} implies that the cost to assign all points $p \in F_i$ to $c_{F_i}$ is equal to
\begin{align*}
\cost(\mathcal{C}) &= \sum_{i \in I} \sum_{p \in F_i} \norm{p - c_{F_i}}^2\\
&= \sum_{i \in I} \bigl[ \sum_{p\in F_i} \norm{p-\mu(F_i)]}^2 + |F_i|\cdot \norm{\mu(F_i) - c_{F_i}}^2\bigr]\\
&= \underbracket{\sum_{i \in I} \sum_{p \in F_i} \norm{p - \mu(F_i)}^2}_{\leq \cost(\mathcal{F})} + \underbracket{\sum_{i \in I} |F_i|\cdot \norm{\mu(F_i) - c_{F_i}}^2}_{\leq \cost(C')}
\end{align*}
In other words, the cost of $\mathcal{C}$ is at most the joint cost of $\mathcal{F}$ and $\mathcal{C}'$.
Since we know that $cost(\mathcal{F})$ is bounded by $\alpha \OPT$ it remains to bound $cost(\mathcal{C}')$, the clustering on the centroids.
To that aim, let us consider an optimum colorblind clustering of $P$ as a reference; denote its set of centers by $C$. 
For each point $p \in P$ let $c(p)$ denote the center in $C$ closest to $p$.\\
For any fixed $p \in F_i$ assigning all copies of $\mu(F_i)$ to $c(p)$ yields a cost of 
\begin{align*}
|F_i|\cdot \norm{\mu(F_i)-c(p)}^2 \le |F_i|\cdot( 2\norm{\mu(F_i)-p}^2 + 2\norm{p-c(p)}^2)
\end{align*}
by the 2-relaxed triangle inequality.
In each fairlet $F_i$, we now look at those points whose distance to $\mu(F_i)$ is at least the median of these distances accross $F_i$.
Formally, we define for each $F_i \in \mathcal{F}$:
\begin{align*}
  N_i := \Bigset{ p \in F_i }{ \exists\ \Bigl\lfloor\frac{|F_i|}{2}\Bigr\rfloor \text{ points $q \not=p$ with } \norm{p - \mu(F_i)}^2 \le \norm{q - \mu(F_i))}^2}
\end{align*}
and let $\bar{p}_i := \arg\min_{p \in N_i} \norm{p - c(p)}^2$ be a point with minimum distance to its center in $C'$, among all points in $N_i$.
With that we obtain $|F_i|\cdot \norm{\mu(F_i) - \bar{p}_i)}^2 \le 2\sum_{p \in F_i} \norm{\mu(F_i) - p}^2$ and $|F_i|\cdot \norm{\bar{p}_i - c(\bar{p}_i)}^2 \le 2\sum_{p \in F_i} \norm{p - c(p)}^2$ for all $i\in I$. 
This implies that there exists a clustering on the centroids with a cost of at most 
\[\sum_{i\in I} \Bigl(4\cdot\sum_{p\in F_i} \norm{\mu(F_i) - p}^2\Bigr) + \sum_{i \in I}\Bigl( 4\cdot\sum_{p \in F_i} \norm{p - c(q)}^2\Bigr) .\]
The first sum is at most $4$ times the cost of $\mathcal{F}$ and therefore bounded by $4 \alpha\OPT$.
The second sum is at most $4$ times the cost of an optimal (colorblind) clustering and is therefore bounded by $4\OPT$.
This shows that there exists a clustering on the centroids with a cost of at most $4(\alpha +1)\OPT$. 
The $\gamma$-approximation therefore has a total cost of at most $\gamma\cdot 4(\alpha +1)\OPT$ and the computed clustering on $P$ has cost of at most $(4\gamma(\alpha + 1) + \alpha)\OPT$.
Therefore, it is a $(4\gamma(\alpha + 1) + \alpha)$-approximation.

\subsection{Computing fairlet decompositions}

Having reviewed fairlet based black box approximations for the different clustering objectives in the previous section, we now turn to the implementation. 
What concrete approximation ratios do we obtain once we fill in the black boxes with the state-of-the-art?

\subsubsection{The happy world of \kk-center\label{appendix:sec:knownresults}}
The colorblind variant of the \kk-center problem has a 2-approximation algorithm~\cite{G85,HS86} and thus, it remains to find an $\alpha$-approximative fairlet decomposition: Together with the results from the previous section, we then obtain a $(2+\alpha)$-approximation for the fair $k$-center problem.
All of this section is about known results, or at least algorithms that achieve the same guarantee as the known results. The prominent feature of the $k$-center problem is the \emph{threshold graph}. Observe that for $k$-center, the optimum \emph{cost} is a pairwise distance: It is the maximum distance between a center (which is an input point itself) and the furthest away point assigned to it. Thus, we can guess the optimum cost by iterating through all $\Theta(n^2)$ pairwise distances. \lq Guessing\rq\ means that we try the $\Theta(n^2)$ values, assume in each run that our guess is the optimum and compute a solution based on this assumption. Depending on the actual algorithm, some runs may fail (indicating that we had the wrong value), and the other runs will give solutions of different quality; we can then pick the best solution, which can only be better than the solution of the run where we indeed had guessed the value correctly. This trick is used in many papers on $k$-center. The usual way to make use of the guessed optimum value -- which we name $\tau$ -- is to build a threshold graph $G_\tau$ where points are connected iff they are at distance $\le \tau$ and then observe that two points can only be in the same optimum cluster if they are connected by a path of length two in $G_\tau$, or, equivalently, if they are connected in $G_\tau^2$.

\paragraph{Two colors.} Now let's start with the case that is studied in~\cite{CKLV17}, i.e., we assume that $|Col|=2$, say $Col=\{R,B\}$. We can then easily define the \emph{balance} of a set $Q$ as
\[
\balance(Q) = \min \left\{ \frac{|R(Q)|}{|B(Q)|},\frac{|B(Q)|}{|R(Q)|}\right\}\in [0,1],
\]
and the balance of a clustering as the smallest balance of any cluster in it. 
Now the main question is what balance we want to achieve. In~\cite{CKLV17}, two cases are distinguished: a) the input $P$ has balance $1$, and we want to compute a clustering which also has exactly balance $1$ and b) we want a clustering with balance $1/t$, but we do not know anything about the balance of $P$. If $\balance(P) < 1/t$, there is no feasible solution, and if $\balance(P) > 1/t$, then we allow additional imbalance.

When considering the difficulty of the approximation problem, it makes sense to distinguish three different cases. For the purpose of later referencing, we state them in the following definition.
\begin{definition}\label{def:twocolor:cases}
 For the two color case, we distinguish the following cases of a fair clustering. These are the cases where we want a clustering
\begin{enumerate}
\item with balance $\balance(P)$, and $\balance(P)=1$,\label{case-one}
\item with balance $\balance(P)$, and $\balance(P)=1/t$ for an integer $t$, \label{case-two}
\item with balance $\balance(P)$, and $\balance(P)=s/t$ for integers $s \le t$, \label{case-two-b}
\item with balance $b \le \balance(P)$, and $b=1/t$ for an integer $t$, or\label{case-three}
\item with balance $b \le \balance(P)$. \label{case-four}
\end{enumerate}   
\end{definition}

In cases~\ref{case-three} and~\ref{case-four}, the target balance $b$ may in particular be smaller than $\balance(P)$. 
Notice that for any finite point set, $\balance(P)=s/t$ for some integers $s,t$, so the real restriction in case~\ref{case-two-b} is that we want to match $\balance(P)$ exactly.
The cases~\ref{case-one},~\ref{case-two},~\ref{case-two-b} involve \emph{exact} fair clustering. Notice that these models force us to exactly match the proportion of red and blue points of $P$. Under this condition, it might not even be possible to find $k$ clusters; thus, the optimal solution may well just consist of $P$. This trivial clustering, however, is always a feasible solution, so the problem is well-defined.

\textbf{{A first algorithm for case~\ref{case-one}.}} 
In case~\ref{case-one}, a good fairlet decomposition is particularly easy to obtain: It consists of a matching between the red and the blue points. 
Assume that $C_1,\ldots,C_k$ is the optimal fair clustering of $P$ with radius $r$. Then in particular, each $C_i$ consists of an equal number of red and blue points. Thus, there exists  a bipartite matching between the red and blue points where both endpoints of a matching edge are always in the same optimum cluster. 

The most intuitive idea is to recover this matching (or a better one) by considering the pairwise distances. Since we consider $k$-center in this paragraph, we guess the optimum cost $\tau$ as described above. Under the assumption that we guessed $\tau$ correctly, we know that any two points in the same optimum cluster are at distance at most $2\tau$. Thus, we just look for a perfect matching in $G_\tau^2$. If none exists, then $\tau$ was wrong. So it suffices to find the smallest $\tau$ for which we succeed. 
Now define a fairlet $F=\{r_i, b_i\}$ with center $r_i$ for each edge $e=\{r_i, b_i\}$ in the perfect matching $M$.
Since $e$ is part of $G_\tau^2$, the distance between $r_i$ and $b_i$ -- and thus the radius of $F$ -- is at most $2\tau = 2\OPT$, where $\OPT$ denotes the cost of an optimum fair clustering as before.
This means that our fairlet decomposition is 2-approximate since its cost is bounded by $2\OPT$.
In total, we obtain a $4$-approximation for the fair $k$-center problem. 

\textbf{An improved algorithm for case~\ref{case-one}.}
Chierichetti et al.~\cite{CKLV17} claim a $3$-approximation for $k$-center, yet the above approach only yields a $4$-approximation guarantee. 
We now develop an improvement. 
Again, partition the points into pairs of one red and one blue point from the same optimal cluster. 
Let $r, b$ be such a pair. 
Observe that then there is a point $c$ (a center in the optimum solution) such that $\dist(r,c)\le OPT$ and $\dist(b,c)\le OPT$. 
Thus, instead of looking at the pairwise distances between all red and blue points, we do something else. 
For each pair $r,b$, we compute the point $x= x(r,b) = \arg \min_{x \in P} \max\{\dist(r,x),\dist(b,x)\}$. 
Also, we set $c(r,b)=\max\{\dist(r,x),\dist(b,x)\}$. 
Then we know that if $r$ and $b$ are in the same optimum clustering, then $c(r,b) \le OPT$.
We build the threshold graph $(G_\tau')^2$ slightly differently now by including an edge for every pair $r,b$ with $c(r,b) \le \tau$. 
A perfect matching in this graph exists iff $\tau \ge OPT$. 
Thus, we search for the smallest $\tau$ such that we find a perfect matching in $(G_\tau')^2$ and then know that in this matching, every point is at distance $\le \tau$ from his partner. 
We again create a fairlet $F=\{r,b\}$ for each edge $e=\{r,b\}$ in the matching; this time, however, we assign $c(r,b)$ as the center of $F$ in the decomposition. 
Since now every fairlet has a radius of at most $\tau=\OPT$, the decomposition is exact and we obtain a $3$-approximation overall. We can apply the same approach to the fair $k$-supplier problem when for each pair $r,b$, we compute the point $x= x(r,b) = \arg \min_{x \in L} \max\{\dist(r,x),\dist(b,x)\}$ and set $c(r,b)=\max\{\dist(r,x),\dist(b,x)\}$. Since the computed decomposition is then exact and we there exist $3$-approximation algorithms for the standard $k$-supplier problem~\cite{HS86}, we obtain a $5$-approximation.

\textbf{Cases \ref{case-two} and~\ref{case-three}.}
In the happy $k$-center world, cases \ref{case-two} and \ref{case-three} of Definition~\ref{def:twocolor:cases} can both be $4$-approximated: Chierichetti et al.~\cite{CKLV17} give a $4$-approximation by using a minimum cost flow. 
We describe a slightly simpler algorithm which assumes that we know the majority color (or compute it as a first step). Without loss of generality, assume that blue is the majority color, i.e., at least half of the points are blue.

The difference to case \ref{case-one} is that the input can no longer be broken into \emph{pairs} of red and blue points. However, we still know that the points can be partitioned into groups of one red and $\le t$ blue points. More precisely, we know: In any clustering satisfying the fairness constraint, a cluster $C$ with $b(C)$ blue points must have at least $\lceil b(C)/t\rceil$ red points. Thus, we can take the optimum solution and partition each optimum cluster in subgroups of one red point together with $\le t$ blue points. Now these subgroups are exactly the fairlets we want to find.

To obtain the partitioning, we set up a flow network built upon the threshold graph. We start by adding all red points as vertices and all blue points as vertices. We add threshold edges between red and blue points in the same manner as in the easy algorithm: For every pair $r,b$, we add an edge if $\dist(r,b) \le 2\tau$ (i.e., if they are connected in $G_\tau^2$). This edge gets a capacity of one. Next, we add a source and a sink. The source is connected to all red points with an edge of capacity $t$: This is the maximum number of points it can facilitate. The blue points are connected to the sink by an edge of capacity $1$. We then compute a maximum flow in this graph. Let $n_b$ be the number of blue points. If the flow has a value less than $n_b$, we know that $\tau$ is wrong: For $\tau \ge OPT$, we know by the above argumentation that we can group the points accordingly and define a flow where every blue point is facilitated by some red point. So we compute the smallest $\tau$ for which the maximum flow has value $n_b$ (it cannot be higher since that is the capacity of the edges going into the sink). Using the red points as centers, this step gives us fairlets with radius $\le 2 \OPT$ and we obtain a $4$-approximation. For $k$-supplier we can again use the same approach. As we already use some of the points as centers we obtain a $5$-approximation.

\textbf{Case~\ref{case-two-b}.}
If $\balance(P)=s/t$, then the above maximum flow idea fails since we lose the anchor node that collects the points of each fairlet. 
However, in the happy world of $k$-center, we can use capacitated clustering to fill the gap. 
The general idea is simple (use capacitated clustering to find fairlets), however, the execution is trickier than one might expect.
The following is a summary of the approximation algorithm given in Section~4.2 in~\cite{RS18}.
The algorithm in ~\cite{RS18} works for multiple colors, but in this section, we only discuss two color variants.

Again, we assume that blue is the majority color, which means that $|R(P)|=n_f \cdot s$ and $|B(P)|=n_f \cdot t$ for some integers $n_f,s,t$ with $gcd(s,t) = 1$. 
We consider an optimal solution. 
Then each cluster in this solution has $j\cdot s$ red and $j \cdot t$ blue points for $j \in \mathbb{N}$, and we can group it into fairlets with $s$ red and $t$ blue points. 
In particular, this grouping induces a uniform capacitated clustering of the red points into $n_f$ clusters with capacity $s$, and a uniform capacitated clustering of the blue points into $n_f$ clusters with capacity $t$, and the radius of both these clusterings is at most $\OPT$.

There is one subtlety, though: the centers of these capacitated clusterings are not necessarily of the same color. 
In the optimum fair clustering, the center of each cluster may be different, sometimes red, sometimes blue. 
This means that the capacitated clustering problems that are induced by the optimum solution are not uniform capacitated \emph{$k$-center} solutions, but they are uniform capacitated \emph{$k$-supplier} solutions: 
For the $k$-supplier problem, the centers do not have to be input points, instead, the input comes with a set of points $P$ and a set of possible center locations $L$ as we indicated at the beginning of the preliminaries. 
The capacitated $k$-supplier problem now asks to minimize $\max_{x \in P} d(x,\phi(x))$ by choosing a set $C \subset L$, $|C|\le k$ and an assignment of the points in $P$ to the centers in $C$ that respects the capacities of the centers in $C$.

Thus, we know that the optimum solution to the following uniform capacitated $k$-supplier problem costs at most $OPT$: Let $P$ be the red points, let $L$ be the union of red and blue points, and set the capacity of every center to $s$. Thus, we use a subroutine for this problem. In the following, we will be okay if the clustering uses soft capacities, i.e., it may open centers multiple times. A solution for this problem can only be cheaper, and we won't be using the centers as centers in the final solution anyway, so the relaxation does not hurt us.

Khuller and Sussmann~\cite{KS00} give a $5$-approximation for the uniform soft capacitated $k$-center problem. Unfortunately, we need the $k$-supplier version. Taking a $k$-center solution for a $k$-supplier version can at most double the approximation factor; thus, we use Khuller and Sussmann's algorithm, but it gives a $10$-approximation for us. The smarter way would be to \emph{adapt} Khuller and Sussmann's algorithm to work for the $k$-supplier variant. We conjecture that this might give a $7$-approximation instead of the $10$.

We apply the algorithm to the red points and obtain clusters of size $s$. 
Now we need to assign $t$ blue points to each red cluster. 
This, however, can be done by computing a matching. 
We construct a bipartite graph. 
On the red side, it contains $t$ copies of each center, i.e., in total, it contains $t \cdot n_f = b(P)$ red points. 
The blue side just consists of the blue points. 
For the edges, we again use thresholding. 
For threshold $\tau$, we add an edge between a red point and a blue point if their distance is at most $\tau$. 
Now we compute the smallest $\tau$ for which the resulting graph has a perfect matching. 
This matching assigns one blue point to every center-copy, resulting in $t$ blue points assigned to every red cluster. 

What's the quality of the resulting solution? 
We know something like this:
We know that in the optimum solution, there is some other clustering of the red points which is coupled with a clustering of the blue points. 
Now if we had the correct red clustering, the distance between a blue point and all red points in the same optimum cluster would be at most $2\cdot\OPT$, leading to an overall approximation ratio of $12$. 
We can still get a similar statement by using Hall's theorem. 
The details of this and a overall formally more careful variant of this proof is in Theorem 22.

This finally gives us fairlets with $s+t$ points and a corresponding center. 
The radius of each fairlet is at most $12\cdot\OPT$ and we obtain a $14$-approximation in total.
For the $k$-supplier version we obtain a $15$-approximation.

\textbf{Case~\ref{case-four}.}
Even in the happy $k$-center world no constant factor approximation algorithm is known for case~\ref{case-four}; in particular, we do not know how to find a good fairlet decomposition.
The problem here is that we do not even know that the optimal clustering can be partitioned into small balanced subsets. 

\paragraph{Multiple colors}
We assume without loss of generality that $|col_0(P)| \le |col_i(P)|$ for all $1 \le i \le g$. Rösner and Schmidt~\cite{RS18} showed that case~\ref{case-two} with two colors and $balance(P) = 1/t$ for an integer $t$ can be generalized to an arbitrary number of colors with $\frac{|col_i(P)|}{|col_0(P)|} \in \mathbb{N}$ for all $1 \le i \le g$. 
We still know that the points in each optimal cluster can be partitioned into groups of one point with color $col_0$ and $\frac{|col_i(P)|}{|col_0(P)|}$ points with color $col_i$. 
The idea is to assign points with color $col_h \in Col \setminus\{col_0\}$, independently of points with a different color, to the points with color $col_0$. In the end all points connected to the same point in $col_0(P)$ build a fairlet.
To do so Rösner and Schmidt~\cite{RS18}  set up a flow network analogously to case~\ref{case-two}: 
It is built upon the threshold graph for each color $col_i \in Col \setminus\{col_0\}$ which contains all points with the colors $col_0$ or $col_i$. 
They then chose the smallest threshold for which the corresponding networks for all colors $col_i \in Col \setminus\{col_0\}$ are successful.
This again results in a fairlet decomposition with radius $\le 2 OPT$ and a $4$-approximation for the fair $k$-center problem as well as a $5$-approximation for the fair $k$-supplier variant.

The case~\ref{case-two-b} with $r > 1, b > 1$ can similarly be generalized to instances with arbitrary many colors~\cite{RS18}.
Let $r = \frac{|col_0(P)|}{gcd(|col_0(P)|,\ldots,|c_{g}(P)|)}$, then we know that the points in each optimal cluster can be partitioned into groups of $r$ points with color $col_0$ and $\frac{|col_i(P)|}{gcd(|col_0(P)|,\ldots,|col_{g}(P)|)}$ points with color $col_i$. The approach is then again to compute a clustering on the points with color $col_0$ in which every cluster contains exactly $r$ points. Dealing with the other colors independently of each other the same approach as in case~\ref{case-two-b} can be used to match $\frac{|col_i(P)|}{gcd(|col_0(P)|,\ldots,|col_{g}(P)|)}$ points with color $col_i$ to each of these sets of $r$ points with color $col_0$.
Again, with the $5$-approximation algorithm for the capacitated \kk-center problem \cite{KS00} this creates a fairlet decomposition with cost of at most $12 \OPT$ and $14$ and $15$-approximations for fair $k$-center and fair $k$-supplier.

Theoretically, case~\ref{case-three} can also be generalized to a case where we require that $\frac{|col_0(P')|}{|col_i(P')|} \ge 1/t_i$ for some integer $t_i$ for all $1 \le i \le g$. This way every cluster in the optimal solution can be partitioned into fair subsets out of which each contains exactly one point with color $col_0$ and at most $t_i$ points of color $col_i$. The generalization of the approach is then identical with the generalization of case~\ref{case-two}. However, this generalization can not be described through our ratio based problem definition with $l,u$-balanced clusters, as in the ratio based problem definition the points with different colors can not be treated independently of each other.

\subsubsection{Fairlet decompositions for \kk-median, facility location and the \kk-means problem}
We now turn to computing fairlet decompositions for clustering objectives beside $k$-center and assume that we want to compute exact fairlets (i.e., not $l,u$-fairlets) in this section.
\paragraph{Instances with two colors}
Let us first consider the simple case where we only have two colors, i.e. $Col=\{red,blue\}$.
In that case we have $\balance(P)=r/b$ for two integers $r,b \in \mathbb{N}$ and $\gcd(r,b)=1$.

\textbf{The case \texorpdfstring{$r=b=1$}{r=b=1}.}
If the balance is $1$, then each fairlet consists of exactly one blue and one red point.
We can therefore model the fairlet decomposition problem as a matching problem as before:
We construct a complete bipartite graph in which the red and blue points make up one partition, repectively.
The cost of an edge between a red point $p$ and a blue point $q$ is the distance between $p$ and $q$. 
For $k$-median/facility location, this cost is equal to assigning both $p$ and $q$ to $p$.
Assigning $p$ and $q$ to any other center cannot yield lower costs by the triangle inequality.
In the $k$-means case, we compute the cost of assigning both $p$ and $q$ to each possible center $c$ and then define the cost of $\{p,q\}$ as the minimum, i.e. $\cost(\{p,q\}) = \min_{c \in L} \norm{p-c}^2 + \norm{q-c}^2$.
We then compute a minimum cost perfect matching to obtain a fairlet decomposition; the cost of the matching is the cost of the decomposition. 
Unfortunately, this is the only case leading to a constant factor approximation algorithms so far.

\textbf{The case \texorpdfstring{$r=1,b > 1$}{r=1, b > 1}.}
If the balance is $1/b$ for an integer $b$, we can solve the fairlet decomposition problem similarly: 
Analogously to the $k$-center case, we start with the $n_b$ many blue points and cluster them into sets with exactly $b$ points each. 
This, again, could be done with an approximation algorithm for the uniform capacitated problem (i.e., \kk-median and \kk-means).
Unfortunately, we are not aware of any such algorithm.
Still, if we set $k$ to $n_b/b$ and the uniform capacity to $b$, then any $\gamma$-approximation to the capacitated problem has cost of at most $\gamma \OPT$, as observed before. 
We can then collect the center of each of the clusters of $b$ blue points, and similarly to the case with $r=b=1$ compute a perfect matching between the red points and these centers. 
Hall's theorem then shows that it is possible to assign one red point $r$ to every blue cluster while guaranteeing that at least one of the blue points is in the same cluster with $r$ in the optimal fair solution. 
For $k$-median/facility location the cost incurred by $r$ is then at most its cost in the optimal solution, plus the cost incurred by its assigned blue point in the optimal solution, plus the cost incurred by the assigned blue point in the computed approximation. 
This gives us a total cost of at most $(2\alpha +1)\OPT$. 
For $k$-means the cost of the red point is then at most four times its cost in the optimal solution, plus four times the cost of its assigned blue point in the optimal solution, plus four times the cost of the assigned blue point in the computed approximation. 
This yields a total cost of at most $4(2\alpha +1) \OPT$.

\textbf{The case \texorpdfstring{$r > 1, b >1$}{r>1,b>1}.}
We further extend the previous case. 
Without loss of generality we assume $b \ge r$, i.e. that blue is the majority color.
Using an approximation for capacitated colorblind clustering, we again cluster the blue points into sets of size $b$ and then create $r$ copies of each center. We can then compute a minimum cost perfect matching between the red points and the centers.

Similar to before Hall's theorem shows that there exists an injective assignment of red points to a blue point in the same optimal fair cluster such that exactly $r$ red points are assigned to the blue points in every blue cluster. 
For $k$-median/facility location the cost of the red point is then at most their cost in the optimal solution plus the cost of their assigned blue point in the optimal solution plus the cost of the assigned blue point in the computed approximation. This gives us a total cost of at most $(2\alpha +1)\OPT$. 
For $k$-means the cost of the red point is then again at most four times this term, which makes the total cost at most $4(2\alpha +1)\OPT$.

\textbf{Instances with arbitrary many colors}
Let $r_i := \frac{|col_i(P)|}{gcd(|col_0(P)|,\ldots,|col_{|Col|}(P)|)}$ be the number of points with color $col_i$ each fairlet should have.
We choose one of the colors $col_i$ and again start by clustering the points into sets of size $r_i$. For each other color $col_j$ ($i \neq j$) we then create $r_j$ copies of each center and compute a minimum cost perfect matching between the points with color $c_j$ and the centers.
We do this for every possible choice of color $col_i$ and then take the best solution that we obtained.
Again Hall's theorem shows that there exists an assignment of points with a color other than $col_i$ to a point with color $col_i$ in the same optimal fair cluster with the following property: Exactly $r_j$ points with color $col_j$ are assigned to the points with color $col_i$ in every cluster, and the number of points assigned to a point with color $col_i$ is either $\left\lfloor \frac{n - |col_i(P)|}{|col_i(P)|}\right\rfloor$ or $\left\lceil \frac{n - |col_i(P)|}{|col_i(P)|} \right\rceil$. For $k$-median/facility location the cost of a point with a color other than $col_i$ is then at most its cost in the optimal solution, plus the cost of its assigned point with color $col_i$ in the optimal solution, plus the cost of the assigned point in the computed approximation.
Assuming that $col_i$ is the color whose points have the smallest average cost in the optimal fair clustering this implies that the total cost is at most $(3\alpha +3) \OPT$. 
For $k$-means the cost of a point with a color other than $col_i$ is then at most four times its cost in the optimal solution, plus four times the cost of its assigned point with color $col_i$ in the optimal solution, plus four times the cost of the assigned point in the computed approximation.
We again assume that $col_i$ is the color whose points have the smallest average cost in the optimal fair clustering which implies that the total cost is at most $4(3\alpha +3) \OPT$.

\end{document}